\documentclass{scrartcl}

\usepackage{algorithm}
\usepackage{algpseudocode}
\usepackage{amsmath}
\usepackage{amssymb}
\usepackage{amsthm}
\usepackage{booktabs}
\usepackage{color}
\usepackage{graphicx}
\usepackage{natbib}
\usepackage{subcaption}
\usepackage{tikz}
\usepackage{tikz-qtree}
\usepackage{cancel}
\usepackage{multirow}

\usetikzlibrary{bayesnet}
\RequirePackage[colorlinks,citecolor=blue,urlcolor=blue]{hyperref}

\newtheorem{theorem}{Theorem}

\newtheorem{proposition}{Proposition}
\newtheorem{lemma}{Lemma}
\newtheorem{corollary}{Corollary}

\newtheorem{remark}{Remark}
\newtheorem{example}{Example}
\newtheorem{model}{Model}

\newcommand{\x}{\bv}
\newcommand{\X}{\bV}
\newcommand{\n}{N}
\newcommand{\rr}{{(r)}}

\newcommand{\ceil}[1]{\left \lceil #1 \right \rceil }

\newcommand{\bigzero}{\mbox{\normalfont\Large\bfseries 0}}

\newcommand{\RR}{\ensuremath{\mathbb{R}}}

\newcommand{\bb}{{\boldsymbol{b}}}

\newcommand{\bg}{{\boldsymbol{g}}}

\newcommand{\bm}{{\boldsymbol{m}}}

\newcommand{\bp}{{\boldsymbol{p}}}

\newcommand{\br}{{\boldsymbol{r}}}

\newcommand{\bv}{{\boldsymbol{v}}}
\newcommand{\bw}{{\boldsymbol{w}}}
\newcommand{\bx}{{\boldsymbol{x}}}
\newcommand{\by}{{\boldsymbol{y}}}
\newcommand{\bz}{{\boldsymbol{z}}}

\newcommand{\bA}{{\boldsymbol{A}}}
\newcommand{\bB}{{\boldsymbol{B}}}
\newcommand{\bC}{{\boldsymbol{C}}}
\newcommand{\bD}{{\boldsymbol{D}}}
\newcommand{\bE}{{\boldsymbol{E}}}

\newcommand{\bI}{{\boldsymbol{I}}}

\newcommand{\bL}{{\boldsymbol{L}}}
\newcommand{\bM}{{\boldsymbol{M}}}

\newcommand{\bQ}{{\boldsymbol{Q}}}

\newcommand{\bT}{{\boldsymbol{T}}}
\newcommand{\bU}{{\boldsymbol{U}}}
\newcommand{\bV}{{\boldsymbol{V}}}

\newcommand{\beeta}{{\boldsymbol{\eta}}}
\newcommand{\bzeta}{{\boldsymbol{\zeta}}}

\newcommand{\bkappa}{{\boldsymbol{\kappa}}}
\newcommand{\bth}{{\boldsymbol{\theta}}}

\newcommand{\bPhi}{{\boldsymbol{\Phi}}}
\newcommand{\bOmega}{{\boldsymbol{\Omega}}}

\newcommand{\bigo}[1]{{\operatorname{\mathcal{O}}\left(#1\right)}}

\newcommand{\law}{\mathcal{L}}

\DeclareMathOperator{\Cost}{Cost}
\DeclareMathOperator{\Chol}{Chol}
\DeclareMathOperator{\Nor}{\mathcal{N}}

\DeclareMathOperator{\Bin}{Binomial}

\newcommand{\red}[1]{{\color{red} {#1}}}

\renewcommand{\r}{{(r)}}

\title{\textbf{Conjugate gradient methods for high-dimensional GLMMs}}
\author{Andrea Pandolfi\thanks{Bocconi University, Department of Decision Sciences},\,
    Omiros Papaspiliopoulos\thanks{Bocconi University, Department of Decision Sciences and BIDSA},\, and Giacomo Zanella\thanks{Bocconi University, Department of Decision Sciences and BIDSA\\ GZ acknowledges support from the European Research Council (ERC), through StG “PrSc-HDBayLe” grant ID 101076564.}
  }

\begin{document}
  \maketitle
\begin{abstract}
Generalized linear mixed models (GLMMs) are a widely used tool in statistical analysis. The main bottleneck of many computational approaches lies in the inversion of the high dimensional precision matrices associated with the random effects. Such matrices are typically sparse; however, the sparsity pattern resembles a multi partite random graph, which does not lend itself well to default sparse linear algebra techniques. Notably, we show that, for typical GLMMs, the Cholesky factor is dense even when the original precision is sparse. We thus turn to approximate iterative techniques, in particular to the conjugate gradient (CG) method. We combine a detailed analysis of the spectrum of said precision matrices with results from random graph theory to show that CG-based methods applied to high-dimensional GLMMs typically achieve a fixed approximation error with a total cost that scales linearly with the number of parameters and observations.
Numerical illustrations with both real and simulated data confirm the theoretical findings, while at the same time illustrating situations, such as nested structures, where CG-based methods struggle.
\end{abstract}
\noindent%
\vfill 

\newpage





\section{Introduction}
Generalized linear mixed models (GLMMs) are a foundational statistical tool, widely used across multiple disciplines \citep{book:gelman, book:wood2017}.
GLMMs extend the framework of Generalized Linear Models by incorporating both fixed effects and random effects. Fixed effects capture population-level trends and relationships, while random effects account for individual deviations from these trends. 
GLMMs utilize a link function to establish a relationship between the mean of the response variable and the linear predictor, which is given by a linear combination of both fixed and random effects (see Section \ref{sec:crossed}).

In various settings, the factors associated with the random effects may include numerous categories, known as \emph{levels}, resulting in models with a large number of parameters $p$ and a large number of observations $\n$, both potentially in the order of several thousands. 
This scenario often arises in contemporary applications. For instance, within the political sciences, one may consider a geographic factor with numerous units, or include factors encoding so-called deep interactions \citep{ghitza2013}. Such models are also applicable in recommendation systems, where categorical variables denote customers and products \citep{GaoOwen2017}.

Our work is relevant for various computational algorithms used to fit GLMMs, such as Gibbs sampling, variational inferences, algorithms for restricted maximum likelihood estimation, and Laplace approximations (see Section \ref{sec:discussion} for a discussion). The primary computational bottleneck of these methodologies lies in the factorization of a sparse high-dimensional matrix, denoted by $\bQ$.  
This matrix decomposes as
\begin{equation}\label{eq:post_prec}
	\bQ = \bT + \X^T \bOmega \X,
\end{equation}
where $\bT$ is the prior component (or, equivalently, a regularization term) and  is usually diagonal or easily factorizable; $\bOmega$ is a diagonal matrix and adjusts for the variances of the response variables; and $\X \in 
\RR ^{\n \times p} $ denotes the design matrix.
In particular, we mostly focus on performing Bayesian inferences using so-called \textit{blocked Gibbs samplers}, that alternate the updates of regression parameters, variance hyperparameters and potential additional latent variables (see Section \ref{sec:glmms}). The update of the regression parameters requires sampling from a multivariate Gaussian distribution with precision matrix $\bQ$ as in \eqref{eq:post_prec}, and it is usually the most computationally intensive steps in those algorithms.

The standard procedure to sample from a multivariate normal distribution involves the computation of the Cholesky factor of $\bQ$ (see Section \ref{sec:chol}). 
It is well established that the Cholesky factorization can be computed efficiently with nested factors \citep[see e.g.][and references therein]{om_giac_note}. 
On the contrary, we will show that, for general crossed factors, under standard random design assumptions, the cost of computing the Cholesky factor scales as $\bigo{p^3}$, even when $\bQ$ is sparse (see Section \ref{sec:negative_SLA} for further details). 
Thus, since exact factorization of $\bQ$ is too expensive, 
we consider an alternative strategy to sample from the desired Gaussian distribution, which involves the solution of a properly perturbed linear system $\bQ \bth = \bb$ (see Section \ref{sec:cg}).  Here, instead of finding the exact solution of the linear system, we employ the well-known conjugate gradient (CG) algorithm to find an approximate solution or, equivalently, produce an approximate sample. 
We will demonstrate that, in the general case where the Cholesky factorization becomes inefficient, the CG sampler only requires a constant (not growing with $N$ and $p$) number of matrix-vector multiplications $\bQ \bb$,
which results in a total of $\bigo{\max(\n, p)}$ cost to produce an approximate sample.

The paper is organized as follows. 
In Section \ref{sec:methodologies}, we briefly review the algorithms for the Cholesky factorization and the CG method, with a particular focus on their time complexity. 
In Section \ref{sec:crossed}, we describe the specific instance of GLMMs, which we refer to as \textit{random-intercept crossed effects models}, that will serve as the reference model for the theoretical analysis in Sections \ref{sec:negative_SLA} and \ref{sec:CG_theory}. 
In Section \ref{sec:negative_SLA}, we show that exact Cholesky factorization is expensive for such models; while, in Section \ref{sec:CG_theory}, we show that the CG algorithm converges fast for such models. 
The bulk of our technical contribution lies in the analysis of the spectrum of $\bQ$ performed in Section \ref{sec:CG_theory}. 
Finally, in Section \ref{sec:numerics}, we describe the proposed methodologies for the case of general GLMMs, and illustrate them with simulated and real data.  
Specifically, we consider an application to a survey data for the American presidential elections of 2004 \citep{ghitza2013}, to a data set for Instructor Evaluations by Students at ETH (from \texttt{lme4} R-library), and to a large-scale recommendation system dataset with $25$ million observations.
The proofs of the presented results are given in the supplementary material, along with a list of symbols intended to facilitate the understanding of the notation. 
The code to reproduce the experiments is available at  \href{https://github.com/AndreaPandolfi/ASLA.jl}{github.com/AndreaPandolfi/ASLA.jl}.

\section{Review of high-dimensional Gaussian sampling}\label{sec:methodologies}
In this section, we review some 
well-known algorithms 
to sample from a high dimensional Gaussian distribution. Specifically, we consider the problem of sampling $\bth \sim \Nor _p(\bQ ^{-1}\bm, \bQ ^{-1})$. Such structure arises in several contexts, such as Bayesian regression \citep{conjugate_gradient_nishimura_suchard}, spatial models with GMRFs \citep{rue2009inla}, and GLMMs, as in this case (see Section \ref{sec:crossed}).
We are mostly interested in the complexity of these  algorithms when $p$ is large and $\bQ$ is sparse. We say that a matrix $\bQ\in \RR^{p\times p}$ is sparse, if the number of non-zero entries grows slower than the number of total possible entries,
namely if	$n_{\bQ} = o(p^2)$ as $p\to +\infty$, 
where $n_\bQ$ denotes the number of non-zero entries of $\bQ$.

\subsection{Cholesky factorization}\label{sec:chol}

The Cholesky factorization of a positive-definite matrix $\bQ\in \RR ^{p\times p}$ is a factorization of the form 
	$\bQ = \bL \bL ^T$, 
where $\bL \in \RR ^{p\times p}$ is a lower triangular matrix with real and positive diagonal entries. Such factorization can be computed with the following column-wise recursion
\begin{equation}\label{eq:chol-rec}
\begin{aligned}
  L_{mm} & = \sqrt{Q_{mm} - \sum_{\ell=1}^{m-1} L_{m\ell} ^2} \,,\qquad
  L_{jm} & = \frac{1}{L_{mm}}\left ( Q_{jm} - \sum_{\ell=1}^{m-1} L_{m\ell} L_{j\ell}\right ), \quad j > m\,.
\end{aligned}
\end{equation}

Given the Cholesky factor $\bL$, one can sample from $\Nor (\bQ ^{-1}\bm , \bQ ^{-1})$ by solving linear systems in $\bL $ and $\bL^T$, which can be done efficiently in $\bigo{n_\bL}$ time via forward and backward substitution. A description of the algorithm can be found in the supplementary material. 


For dense matrices, the exact computation of the Cholesky factor requires $\bigo{p^3}$ time, and the storage of $\bigo{p^2}$ entries.
For sparse matrices with specific structures, the computational requirements can be drastically reduced. If $\bQ$ can be turned into a banded matrix via row and column permutations, the cost of computing $\bL$ can be reduced to $\bigo{pb^2}$, where $b$ is the bandwidth. When considering spatial Gaussian Markov random fields (GMRFs) or nested hierarchical models, such reordering can be done efficiently, which makes sampling via Cholesky factor efficient in such cases (see \cite{om_giac_note}, and extensive references therein). 
On the other hand, there is little work in the literature that studies the computational cost of the Cholesky factorization for general GLMMs with non-nested designs.

\subsubsection{Conditional Independence Graphs and Cholesky complexity analysis}\label{sec:cost_chol}
The conditional independence structure of a given Gaussian vector $\bth \sim \Nor _p (\bQ^{-1}\bm, \bQ ^{-1})$ is described by $\bQ$ through the relation
\begin{equation}\label{eq:cond_indep}
 \theta _j\perp \theta _m \mid \bth _{-jm} \Longleftrightarrow Q_{jm} = 0 \,,
 \end{equation} 
where $\bth _{-jm}$ denotes $\bth$ after the removal of the entries $j$ and $m$. 
Hence, the \textit{conditional independence }(CI) \textit{graph} of $\bth$ is entirely described by the support of $\bQ$ \citep[Ch. 2.1.5]{book:rue}. We denote the CI graph of $\bth$ by $G _{\bQ}$, whose vertices are the variables $\{\theta _j, j = 1, \dots, p\}$ and the edges are those $(\theta_j, \theta_m)$ s.t. $Q_{jm}\neq 0$ for $j\neq m$. 

The Cholesky factor $\bL$ of $\bQ$ has also a probabilistic interpretation: for $m<j$, $L_{jm} =0$ if and only if $\theta _j$ and $\theta _m$ are independent given the \textit{future set} of $\theta _m$ excluding $\theta _j$, i.e.\ $\theta _m \perp \theta _j \mid \bth _{\{(m+1):p \}\backslash j}$ \citep[Theorem 2.8]{book:rue}. Therefore, even when $\bQ$ is sparse, $\bL$ might not be and the order we assign to elements of $\bth$ influences the sparsity of $\bL$. 
The possibly non-zero entries of $\bL$ can be deduced from $G_\bQ$ as follows: 
a sufficient condition to ensure $L_{jm} = 0$ for $m<j$ is that the future set of $\theta_m$ separates it from $\theta_j$ in $G_\bQ$. This motivates defining the number of possible non-zero entries in $\bL$ as $n_\bL = \sum_{m=1}^p n_{\bL,m}$ where
\begin{equation}\label{eq:n_Lm}
  n_{\bL,m} = |\{j \geq m : 
  \text{the future set of $\theta_m$ does not separate it from $\theta_j$\ in $G_\bQ$}\}| \,.
\end{equation}
Thus, the Cholesky factor $\bL$ involves $n_\bL - n_\bQ\geq 0$ additional potential non-zero entries compared to the original matrix $\bQ$. Such additional non-zeros terms are commonly referred to as \emph{fill-ins}. 
Since $n_\bL$ depends on the ordering of the variables in $\bth$, standard algorithms for Cholesky factorizations of sparse matrices proceed in two steps: first they try to find an ordering of variables that reduces $n_\bL$ as much as possible, and then compute the corresponding Cholesky factor using \eqref{eq:chol-rec}. Finding the ordering that minimizes $n_\bL$ is NP-hard, but various heuristic strategies to find good orderings are available \citep[Ch. 11]{book:golub2013}. 

Sparsity in $\bL$ has direct consequences on the computational cost required to compute it - although it has to be appreciated that a sparse Cholesky is not necessarily computable efficiently. 
The following theorem quantifies this connection.

\begin{theorem}\label{thm:cost_L}
Denoting with $\Cost(\mathrm{Chol})$ the 
 number of floating point operations (flops)
 needed to compute the Cholesky factor $\bL$ of a positive definite matrix $\bQ$, we have
\begin{equation}\label{eq:costChol_bounds}
  \bigo{n_\bL^2/p} \leq \Cost(\mathrm{Chol}) = \bigo{\sum_{m=1}^{p} n_{\bL,m}^{2}} \leq \bigo{n_\bL^{1.5}}.
\end{equation}
\end{theorem}

The equality and lower bound in \eqref{eq:costChol_bounds} are well-known, while the upper bound is more involved, and we have not been able to find it in the literature. Note that, trivially, the result also implies
  $\Cost(\mathrm{Chol})\geq \bigo{n_\bQ^2/p}$.

When $\bQ$ is a dense matrix we have $n_\bL = \bigo{p^2}$ and thus the lower and upper bounds in \eqref{eq:costChol_bounds} coincide, being both cubic in $p$, and they are both tight. For sparse matrices, instead, the two bounds can differ up to a $\bigo{p^{0.5}}$ multiplicative factor and each can be tight depending on the sparsity pattern. For example, if $\bQ$ is a banded matrix with bandwidth $b$, we have $n_\bL = \bigo{pb}$ and $\Cost(\mathrm{Chol}) = \bigo{pb^2}$, hence the lower bound is tight while the upper bound is off by a $\bigo{(p/b)^{0.5}}$ factor. On the contrary, for a matrix $\bQ$ with a dense $p^{0.5} \times p^{0.5}$ sub-matrix and diagonal elsewhere we have $n_\bL = \bigo{p}$ and $\Cost(\mathrm{Chol}) = \bigo{p^{1.5}}$, meaning that the upper bound is tight while the lower bound is off by a $\bigo{p^{0.5}}$ factor.

\subsection{Conjugate Gradient}\label{sec:cg}

Conjugate Gradient (CG) is a widely-used iterative optimization algorithm employed for solving large systems of linear equations $\bQ \bth = \bb$, for which $\bQ$ is positive-definite \citep{book:golub2013, book:saad2003}. 
CG iteratively minimizes the quadratic form associated to the linear system. However, instead of optimizing along the gradient direction, it restricts to the component of the gradient which is conjugate (i.e.\ $\bQ$-orthogonal) to the previous search directions.

Each CG iteration only requires evaluation of matrix-vector products $ \bQ \bth$ and scalar products of $p$-dimensional vectors. This feature makes CG methods very appealing for solving sparse linear systems, as each iteration only requires $\bigo{n_\bQ}$ operations. 
The algorithm is also optimal in terms of memory efficiency, as it only requires storing $\bigo{p}$ values: basically only the approximate solution, which gets updated in place.
 
 Several strategies have been proposed to use CG algorithm to sample from high dimensional Gaussian distributions. 
In this paper, we will refer to the \emph{perturbation optimization sampler} \citep{papandreou2010, conjugate_gradient_nishimura_suchard}, which requires solving the following linear system
\begin{equation}\label{eq:cg_sampler}
    \bQ \bth  = \bm + \bz, \qquad \bz \sim \Nor (\mathbf{0}, \bQ )
\end{equation}

Simple computations show that, when the linear system in \eqref{eq:cg_sampler} is solved exactly, $\bth$ is an exact sample from $\Nor _p (\bQ ^{-1}\bm,\, \bQ ^{-1})$. If the linear system is solved via CG method, an approximate sample will be returned.
A potential limitation for the application of this algorithm is given by the necessity to sample efficiently from $\Nor _p (\mathbf{0}, \bQ)$, which is however feasible for matrices as in \eqref{eq:post_prec}  (see end of Section \ref{sec:crossed}).
Notice that in this case, CG is only used to solve the linear system in \eqref{eq:cg_sampler}, while other CG samplers \citep{parker_fox_2012, CG_review} exploit CG algorithms to build an approximate low-rank square root of $\bQ ^{-1}$ which is then used to sample. 
Notice, however, that, if $\bQ$ is full-rank but has only  $k<p$ distinct eigenvalues, then \eqref{eq:cg_sampler} would return an exact sample in $k$-steps (see Theorem \ref{thm:cg_conv_rate}), while other methodologies based on low-rank approximations would not. 


\subsubsection{Rate of convergence}
When studying the complexity of CG algorithm, one needs to quantify how fast the solution at iteration $k$ approaches the exact solution of the linear system. 
The convergence behavior of CG algorithm has been extensively studied in the literature \citep[see][Section 11]{book:golub2013}. We report the most well-known results in the following Theorem.
\begin{theorem}\label{thm:cg_conv_rate}
 Consider the linear system $\bQ \bth = \bb $, with $\bQ$ positive definite. Denote the starting vector with $\bth^0$, then the $\bQ$-norm distance between the $k$-th CG iterate $\bth ^k$ and the solution $\bth$ satisfies the inequality
\begin{equation}\label{eq:CG_convergence}
	\dfrac{||\bth ^k  - \bth ||_\bQ}{||\bth ^0  - \bth ||_\bQ} \leq 2 \left  ( \dfrac{\sqrt{\kappa (\bQ)} - 1}{\sqrt{\kappa(\bQ)} + 1} \right )^k ,
\end{equation}
where $\kappa(\bQ) = \lambda_{max}(\bQ) / \lambda_{min}(\bQ)$ denotes the \textit{condition number} of $\bQ$.
Moreover, if $\bQ$ has only $k<p$ distinct eigenvalues, CG returns the exact solution after $k$ iterations.
\end{theorem}
See \cite{book:trefethen} for a proof of these results. Theorem \ref{thm:cg_conv_rate} shows that CG has a fast convergence rate when either $\bQ$ has a small condition number or when it has few distinct eigenvalues. 

It is well-known in the CG literature that the convergence rate presented in Theorem \ref{thm:cg_conv_rate} is very conservative. Indeed, CG is also fast when most of the eigenvalues, except few outlying ones, are clustered in a small interval $[\mu _s , \mu _{p-r}]$, with $\mu_s$ not close to $0$. 
In this case, one observes that, after a small number of iterations, CG behaves as if the components corresponding to the outlying eigenvalues have been removed, and the CG convergence rate changes as if the condition number in \eqref{eq:CG_convergence} is replaced by the effective value  $ \mu_{p-r}/\mu_s$ \citep{book:vorst2003}. This behavior is known as \textit{superlinear convergence of conjugate gradients}. 
For this reason, instead of considering $\kappa (\bQ)$, we will focus on the \textit{effective condition number} \citep{sluis1986}
\begin{equation}\label{eq:cn_eff}
	\kappa _{s+1, p-r}(\bQ ) = \dfrac{\mu _{p-r}}{\mu _{s+1}}\,,
\end{equation}
obtained by removing the smallest $s$ eigenvalues of $\bQ$ and the $r$ largest ones.
In Section \ref{sec:CG_theory}, we will provide upper bounds on $\kappa _{s+1, p-r}(\bQ )$ in the GLMM context, for small $s$ and $r$.

\section{Random-intercept crossed effects models}\label{sec:crossed}
We now describe the model we will refer to in the theoretical analysis of Section \ref{sec:negative_SLA} and \ref{sec:CG_theory}.
Specifically, we present the class of \textit{random-intercept crossed effects models}.
\begin{model}[Random-intercept crossed effects models]\label{mdl:crossed}
For each observation $i\in \{ 1, \dots, \n\}$, consider a univariate continuous response $y_i\in \RR$ distributed as 
\begin{equation}\label{eq:crossed}
	y_{i} \mid \eta_{i} \sim \Nor \left (\eta_{i}, \ \tau ^{-1}\right), \qquad \eta_{i}= \theta_0+\sum_{k=1}^{K} \bz _{i,k}^T\bth _k,
\end{equation}
where $\theta _0$ is a global intercept, 
$\bth _k = (\theta _{k,1}, \dots, \theta _{k, G_k})^T$ is the vector of random effects for factor $k$, and $\bz _{i,k } \in \{ 0, 1\} ^{G_k}$ with $\sum_{g=1}^{G_k} z_{i,k,g}=1$ is a ``one-hot" vector which encodes the level of factor $k$ for the observation $i$.
Here $K$ is the number of factors, and $G_k$ the number of levels in factor $k$. 
We assume independent priors $\theta	_{k,g}\overset{ind.}{\sim} \Nor (0, T_k ^{-1})$, for each factor $k$ and level $g$.
The fixed effect is assigned either a normal or improper flat prior, and the precision parameters are assigned conjugate gamma distributions. 
\end{model}

If we define $\bth = (\theta_0, \bth_1^{T}, \dots, \bth _K ^{T})^{T} \in \RR ^{p}$ (with $p=1 + \sum _k G_k$) and $\x _i = ( 1, \bz_{i,1}^T, \dots, \bz _{i,K}^T)^T \in \RR ^p$, we can write $\eta_i$ as $\x _i ^T \bth$.
Under Model \ref{mdl:crossed}, the posterior conditional distribution of $\bth$ is
\begin{equation}\label{eq:crossed_poterior_distr}
	\bth \mid \by , \X, \bT, \tau \sim \Nor _p  ( \bQ ^{-1} (\bT \bm_0 + \tau \X ^T \by ) ,\  \bQ ^{-1})\,,
\end{equation}
where $\bm_0$ is the prior mean of $\bth$, $\bT$ its prior precision, $\X  = [\x_1 | \dots |\x_\n ] ^T \in \RR ^{\n\times p}$ and 
\begin{equation}\label{eq:crossed_posterior_prec}
	\bQ = \bT + \tau \X ^T \X\,.
\end{equation} 

In Sections \ref{sec:negative_SLA} and \ref{sec:CG_theory}, we study the time complexity of sampling from \eqref{eq:crossed_poterior_distr} with the algorithms presented in Section \ref{sec:methodologies}. Regarding the CG sampler, we will only focus on the complexity of the CG step. Indeed, sampling $\bz \sim \Nor _p (\mathbf{0}, \bQ ) $ with $\bQ$ as in \eqref{eq:crossed_posterior_prec}, can be done in $\bigo{\n + p}$ time, by sampling $\bzeta \sim \Nor _p (\mathbf{0}, \bI _p)$ and $\ \beeta \sim \Nor _\n (\mathbf{0}, \bI _\n ) $ independently, and setting $\bz =  \bT ^{1/2} \bzeta + \sqrt{\tau}\X ^T \beeta$.

Random-intercept models defined as in Model \ref{mdl:crossed} are a specific instance of the more general class of GLMMs considered in Section \ref{sec:numerics}, which motivates our work.
Throughout the paper, Model \ref{mdl:crossed} will serve as a study case to develop theoretical results.
We restrict our attention to the random-intercept case as it allows for an exhaustive theoretical tractability, while still preserving the crossed sparsity structure that constitutes the main computational challenge for general GLMMs with non-nested designs (see e.g.\ Section \ref{sec:negative_SLA}). 
 Notice also that various recent works provided methodologies and theoretical results for this class of models \citep{Biometrika_Om_Giacomo_Roberts, Backfitting_for_crossed_random_effects}.
Note that, while we develop theoretical results in the context of Model \ref{mdl:crossed}, our methodology and numerical results apply to the more general class of GLMMs described in Section \ref{sec:numerics}, which in particular includes random slopes and interaction terms. 
 
For the sake of simplicity, the only fixed effect included in Model \ref{mdl:crossed} is the intercept $\theta_0\in\RR$. 
Nonetheless, all the results of the following sections could be restated to include multivariate fixed effects, $\bth_0 \in \RR ^{D_0}$, leading to the same conclusions, provided that $D_0$ is sufficiently small relative to the dimensionality of the random effects $(\bth _1, \dots, \bth _K)$ (see
Section \ref{sec:suppl_mult_fixed} in the supplementary material for further details).

\subsection{Design assumptions}\label{sec:assumption_notation}
Consider the precision matrix defined in \eqref{eq:crossed_posterior_prec}. Since $\bT$ is a diagonal matrix, the off-diagonal support of $\bQ$ is entirely characterized by the likelihood term, which we will denote with $\bU=\X ^T \X = \sum _{i=1}^n \x _i \x_i ^T$. 
In particular, simple computations show that
\begin{equation}\label{eq:ci_graph}
\begin{aligned}
	U[\theta _0, \theta_0] &= \n \, ,
	&U[\theta _0, \theta _{k,g}] &=  | \{ i = 1,\dots, \n \colon z_{i,k,g}=1\}|\, ,\\
	U[\theta _{k,g }, \theta _{k,g}] &=| \{ i \colon z_{i,k,g}=1\}|\, ,\quad 
	&U[\theta _{k,g }, \theta _{k',g'}] &=  | \{ i  \colon z_{i,k,g}=1,\, z_{i,k',g'}=1\}|\, ,
\end{aligned}
\end{equation}
where $U[\theta _{k,g }, \theta _{k',g'}] $ refer to the entry relative to the pair $(\theta _{k,g }, \theta _{k',g'})$. We assume without loss of generality that each level of each factor is observed at least once, which is equivalent to say that $U[\theta _{k,g }, \theta _{k,g}]\geq 1$, for all factors $k=1,\dots, K$ and levels $g = 1\dots, G_k$.

In the analogy with conditional independence graphs described in Section \ref{sec:cost_chol}, the last equality in \eqref{eq:ci_graph} implies that $\theta _{k,g}$ is connected to $\theta _{k', g'}$ in $G_\bQ$ if and only if $k\neq k'$ and there exists at least one observation where the two levels are observed together. Thus, the resulting CI graph is a $(K+1)$-partite graph, with one trivial block containing the vertex $\theta_0$, and a block of size $G_k$ for each factor $k=1,\dots, K$. Since $U[\theta _0, \theta _{k,g}] = U[\theta _{k,g}, \theta _{k,g}]>0$ for each $k,g$, then $\theta_0$ is connected to all the other vertices.

Throughout our analysis, we will consider a regime where $N,p\to\infty$, while $K$ is fixed. 
Notice that $\n\geq p/(K+1)$, since we assume that each $\theta _{k,g}$ appears in at least one observation.
Finally, we assume that $n_\bQ$ or, equivalently, the number of edges is $G_\bQ$ is  $ \bigo{\n}$.

%

\section{Cholesky factorization with sparse crossed designs}\label{sec:negative_SLA}
In this section, we will provide a theoretical example as well as numerical results with random designs showing that sparse Cholesky factorization is not suitable to efficiently factorize the posterior precision matrices that arise from Model \ref{mdl:crossed}.
\begin{proposition}\label{prop:def_ordering}
Under Model \ref{mdl:crossed},  $n_\bL$ is a non-increasing function of the position of $\theta_0$ in the ordering of $\bth$. 
Hence, placing $\theta_0$ last in the ordering 
always minimizes $n_\bL$.
\end{proposition}

Proposition \ref{prop:def_ordering} implies that we can place $\theta_0$ last without loss of generality.
This helps to reduce $n_\bL$ but it does not solve the problem. The following example describes a simple sparse design that is catastrophic in this sense, where $n_\bQ = \bigo{p}$, but $n_\bL = \bigo{p^2}$.


\begin{example}\label{ex:worst_sla}
Consider a random intercept model with $K=2$ and $G_1 = G_2 = G$. Fix an integer $d \geq 2$, and for each $g = 1, \dots, G-1$, connect the vertex $\theta _{1,g}$ to $\theta_{2,j}$ for all $j$'s for which at least one of the following conditions hold:
\textbf{(a)} $g=j$; \textbf{(b)} $d(g-1) \leq j-2 < dg \mod (G-1)$ and $g<j$; \textbf{(c)} $d(g-1) \leq j-1 < dg \mod (G-1)$ and $g>j$.
For $g=G$ connect $\theta _{1,G}$ to all the $\theta_{2,j}$'s that have degree less or equal than $d$. A graphical representation of the resulting precision matrix $\bQ$ can be found in Section \ref{sec:suppl_figures} of the supplementary material. 
\end{example}

\begin{proposition}\label{prop:unfavourable}
Under the design of Example \ref{ex:worst_sla} with $d$ fixed and $G \to \infty$, it holds that 
\begin{equation}
	n_\bQ = \bigo{pd},\ n_\bL = \bigo{p^2} \text{ and } \Cost(\mathrm{Chol}) = \bigo{p^3}.
\end{equation} 
\end{proposition}

We now provide a numerical study that shows that sparse Cholesky factorization requires $\bigo{p^3}$ time for most sparse crossed designs. 
We consider missing completely at random designs where each cell in the $K$-dimensional data contingency table contains an observation with probability $\pi\in(0,1)$ and is empty otherwise. 
We take $G_1 = \cdots = G_K = G$, hence $\n \sim \Bin(G^K, \pi)$, and we consider, for increasing values of $G$, different values of $K$ and different ways that $\pi$ relates to $G$. Specifically, we examine the following designs:
\[\text{(a) } K=2, \ \pi = 20G^{-1}\,; \qquad\text{(b) } K=2, \ \pi = G^{-1/2}\,; \qquad\text{(c) } K=5, \ \pi = G^{-K+3/2}\,.\]

Notice that the expected degree is constant in scenario (a), but it increases with $G$ in scenarios (b) and (c). 
Moreover, notice that by picking levels independently among all the possible $G^K$ combinations, we guarantee $n_\bQ = \bigo{\n}$.
Figure \ref{fig:ErdosRenyi} reports the results, plotting the number of flops required to compute $\bL$, versus $p = K G + 1 $ in a log-log scale. 
In all three regimes, despite $\bQ$ being sparse, the Cholesky factor $\bL$ is dense, leading to the worst case scenario $\Cost(\mathrm{Chol}) = \bigo{p^3}$. 
All the experiments in the paper are obtained using the approximate minimum degree (AMD) ordering, which proved the most effective in maximizing the sparsity of $\bL$.
We have also explored other ordering strategies which led to the same results. 
In particular, a possible strategy is to order the factors in decreasing order with respect to the number of levels $G_k$. In this case, the first $G_1$ columns have the same sparsity pattern as $\bQ$, and fill-ins are observed only in the remaining submatrix of size $(\sum_{k\neq 1} G_k)^2$  \citep{bates2025}. The above results can also be rephrased by substituting $p$ with $(p-\max_k G_k)$.
In Figure \ref{fig:ErdosRenyi}, we also include the number of flops required for sampling with CG algorithm, which we will analyze in Section \ref{sec:simulated}.

\begin{figure}[h!]
  \centering
  \begin{subfigure}{.325\textwidth}
    \includegraphics[width=\linewidth]{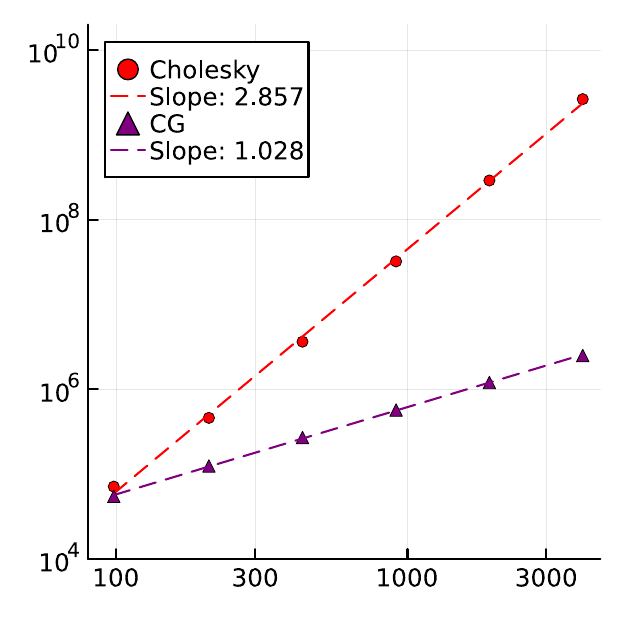}
\end{subfigure}
  \begin{subfigure}{.325\textwidth}
    \includegraphics[width=\linewidth]{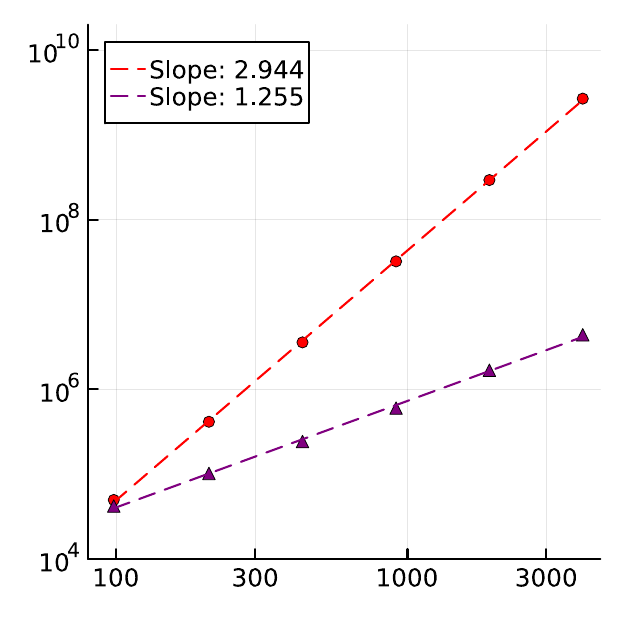}
\end{subfigure}
  \begin{subfigure}{.325\textwidth}
    \includegraphics[width=\linewidth]{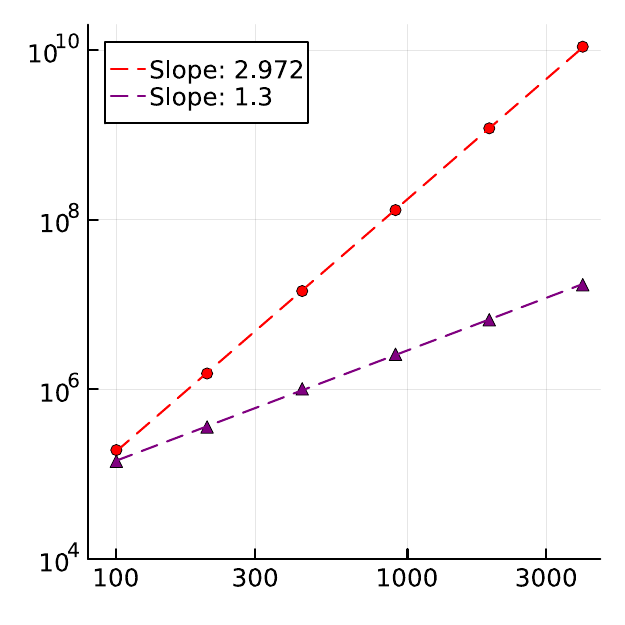}
\end{subfigure}
\caption{$\Cost(\mathrm{Chol})$ (red dots) as a function of $p$ ($x$-axis) in a log-log scale, for scenarios (a)--(c) described above. The purple triangles represent the cost of the CG sampler, which is discussed in Section \ref{sec:simulated}. $\bQ$ is obtained by setting $\tau = 1$ and $\bT = \bI _p$.} 
\label{fig:ErdosRenyi}
\end{figure}

\section{Spectral analysis and CG convergence}\label{sec:CG_theory}

In this section, we provide a series of spectral results about the posterior precision matrix for Model \ref{mdl:crossed}. Namely, we show that all but a fixed number of the eigenvalues of $\bQ$, if appropriately preconditioned, concentrate in a small interval, as $p\to +\infty $. In particular, Theorem \ref{thm:bipartite_biregular} and \ref{thm:suff_cond_strong_connnectivity} will show that the effective condition number remains bounded for increasing $p$. This suggests that the number of CG iterations required to reach a desired level of accuracy is independent of the size of the problem, leading to dimension-free convergence of the algorithm, which is consistent with numerical experiments (see Section \ref{sec:simulated}).


\subsection{Notation}\label{sec:notation}
We denote with $Diag(\bQ)$, the diagonal matrix whose diagonal coincide with the one of $\bQ$ and, for a given vector $\bv$, we denote with $Diag(\bv)$ the diagonal matrix whose diagonal is $\bv$. Recalling that $\bU = \X ^T \X$, we start by decomposing $\bQ$ in \eqref{eq:crossed_posterior_prec} as $ \bQ	 = \bT + \tau \cdot Diag (\bU )\ +\ \tau \bA\,,$
where $\bA = \bU - Diag(\bU) $ can be seen as a weighted adjacency matrix, which counts the number of times that two levels $\theta_{k,g}$ and $\theta _{k',g'}$ are observed together, see \eqref{eq:ci_graph}. We also define $\bD = Diag (\bA \mathbb{1}_p)$, where $\mathbb{1}_p\in \RR ^p$ is a vector of ones, so that $\bD$ is a diagonal matrix whose elements are the row-wise sums of $\bA$. Notice that, for crossed effect models, $\bD= K \, Diag (\bU)$, indeed  $U[\theta _{k,g}, \theta_{k,g}]$ counts the number of times that the level $\theta _{k,g}$ is observed, see \eqref{eq:ci_graph}, however, each time it is observed it is also ``connected'' to other $K-1$ levels from the other factors and to the global parameter $\theta_0$, hence the equality.

When we consider the sub-matrix relative to the random effects only, we use the notation  $\bQ ^\rr\in \RR ^{p-1\times p-1}$, which is obtained by removing the row and column relative to the global effect $\theta_0$. We use the same notation for the adjacency matrix $\bA ^\rr$. We let $\bD ^\rr= Diag(\bA ^\rr\mathbb{1}_{p-1}) \in \RR ^{(p-1) \times (p-1)}$, so that $\bD^\rr= (K-1)\cdot  Diag(\bU^\rr)$.

Denote with $\bar{\bQ}$ the Jacobi preconditioned precision matrix \citep[Sec. 11.5.3]{book:golub2013}, i.e.\
\begin{equation}\label{eq:Q_bar}
	\bar{\bQ} = Diag(\bQ) ^{-1/2}\bQ Diag(\bQ) ^{-1/2}\,.
\end{equation}
In Section \ref{sec:jacobi}, we expand on the importance of such preconditioning. 
We define $\bar{\bA}^\rr= (\bD ^\rr)^{-1/2}\bA^\rr (\bD ^\rr)^{-1/2}$, as one would usually normalize the adjacency matrix of a graph.

\subsection{Jacobi preconditioning}\label{sec:jacobi}
In this section, we briefly investigate the choice of Jacobi preconditioning.
The general idea of preconditioning \citep[Section 11.5]{book:golub2013} is to find a non-singular and easy to factorize matrix $\bM$ s.t.\ the preconditioned matrix $ \bM ^{-1/2} \bQ \bM ^{-1/2}$ is better conditioned than the original $\bQ$. A simple modification of the standard CG algorithm allows working with the preconditioned matrix instead of $\bQ$ \citep[Algorithm 11.5.1]{book:golub2013} at the additional cost of computing $\bx \to \bM^{-1}\bx$ at each iteration. For positive definite matrices, the Jacobi preconditioning is often a natural choice, since $\bx \to \bM ^{-1}\bx$ is very cheap to evaluate. 

For the specific case of Model \ref{mdl:crossed}, the Jacobi preconditioning allows rescaling $\bQ$ so that its eigenvalues lie within a bounded interval independently of the design (see Theorem \ref{thm:outlying_eigvals} for more details). 
To illustrate this point, we consider the following setting. There are $K$ factors of size $G_1, \dots, G_K$ and, for a given $\n$, we sample uniformly without replacement $\n$ out of the possible $\prod_{k=1}^K G_k$ combinations of levels. This generates a specific random design $\X$. 
Notice that this is equivalent to the random design of Figure \ref{fig:ErdosRenyi}, but for general $G_k$ and for fixed $\n$.
In this case, simple computations show that the expected value of the diagonal element $Q[\theta _{k,g}, \theta_{k,g}] $ is $T_k + \tau N/G_k$.
When looking at the spectrum of the resulting $\bQ$, one can observe, for each factor $k$, a bulk of $G_k$ eigenvalues centered around the expected value $T_k + \tau N/G_k$, which is a manifestation of spectral concentration for high-dimensional random matrices.
 The left panel in Figure \ref{fig:spectrum_and_preconditioning} shows this phenomenon in the case of $K=3$ factors. 
On the other hand, when applying Jacobi preconditioning to $\bQ$, the bulks relative to different factors are grouped into a single one centered around $1$ (see the right panel of Figure \ref{fig:spectrum_and_preconditioning}). 
\begin{figure}[h!]
	\centering
    \includegraphics[width=.8\textwidth]{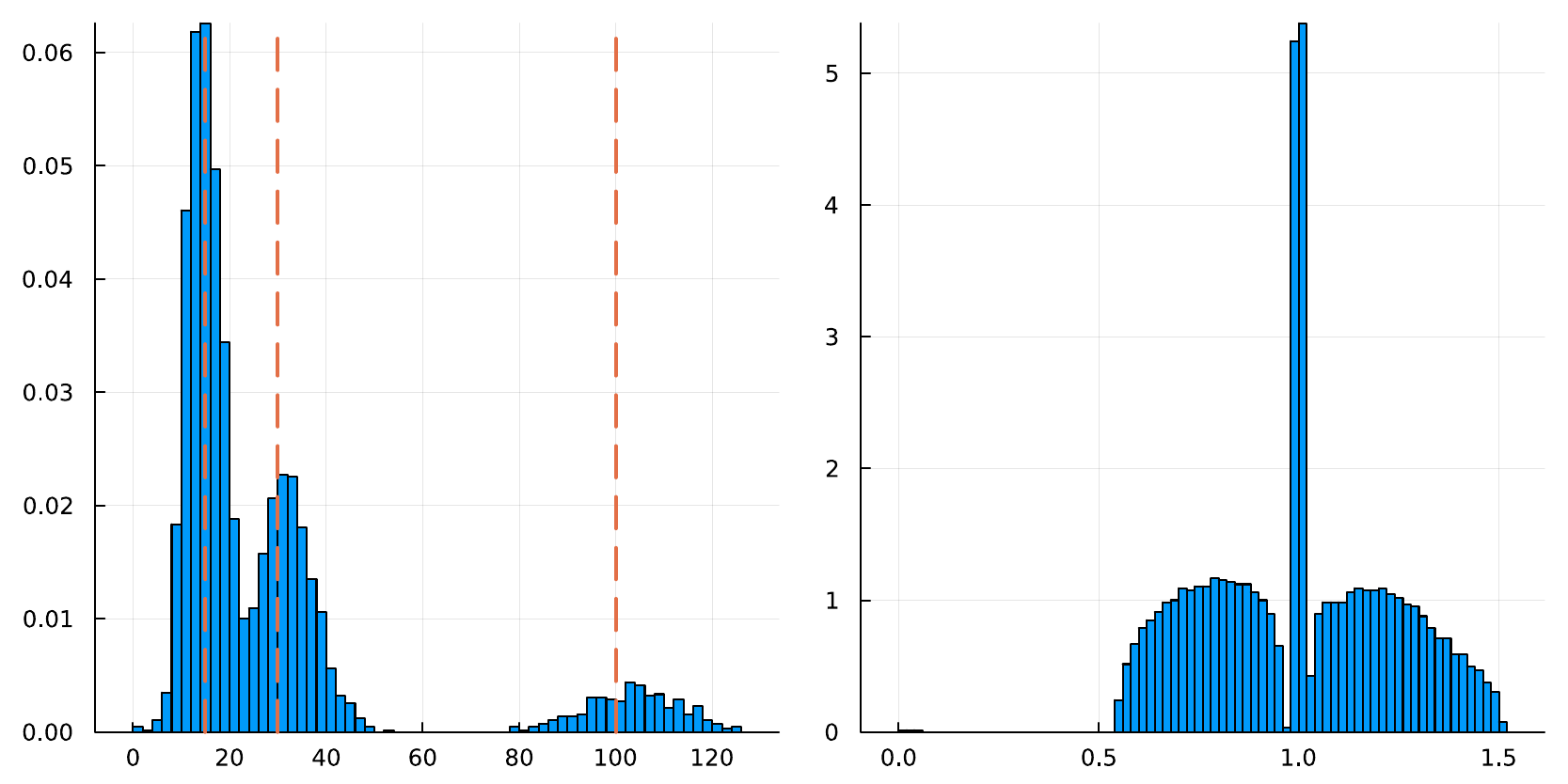}
	\caption{The two histograms show the spectrum of $\bQ$ (left) and $\bar{\bQ}$ (right). They are obtained by setting $\bT= \bI _p, \ \tau = 1$, $G_1 = 300$, $G_2 = 1000$, $G_3 = 2000$. 
The dashed lines on the left panel correspond to the values $T_k + \tau \n /G_k$, for $k=1,2,3$.
}\label{fig:spectrum_and_preconditioning}
\end{figure}
These two different spectral behaviors are reflected in the number of iterations needed for convergence of CG. 
See for example Table \ref{tbl:spectrum_and_preconditioning}, which considers a sequence of designs with $2 $ factors of increasingly different sizes. There, the effective condition number and the number of CG iteration increase as the difference between factor sizes increases, while they remain constant when considering the Jacobi preconditioned version $\bar{\bQ}$.
Theorems \ref{thm:bipartite_biregular} and \ref{thm:suff_cond_strong_connnectivity} provide theoretical results supporting the empirical behavior of the spectrum of $\bar{\bQ}$. 
Given the better performances, we will focus our analysis on the spectrum of $\bar{\bQ}$.

Note that, in the right panel of Figure \ref{fig:spectrum_and_preconditioning}, one can observe a cluster of $2$ small eigenvalues separated from the remaining part of the spectrum. This shows why we need to consider effective condition numbers defined as in \eqref{eq:cn_eff}, and it will be made formal in Theorem \ref{thm:outlying_eigvals}. 

\begin{table}[h!]
\centering

\begin{tabular}{|c|r|r|c|c|c|c|}
  \hline
  \multirow{2}{*}{$G_1$} & \multirow{2}{*}{$G_2$\hspace{.2cm}} & \multirow{2}{*}{$\n$\hspace{.5cm}} & \multicolumn{2}{|c|}{No preconditioning} & \multicolumn{2}{|c|}{Jacobi preconditioning}\\
  \cline{4-7} & &  & $\kappa_{3, p-2}$ & n. iters & $\kappa_{3, p-2}$ & n. iters \\\hline\hline
  100 & 100   & 2828    &  10.61 & 31 & 5.37 & 25 \\
  100 & 1000  & 36482   &  33.08 & 50 & 4.39 & 24 \\
  100 & 10000 & 1015037 & 166.91 & 64 & 3.33 & 19 \\
  \hline
\end{tabular}
\caption{Effective condition number and number of CG iterations for $\bQ$ and $\bar{\bQ}$. The design is the same as in Figure \ref{fig:ErdosRenyi}, with $\pi = p ^{-9/10}$. $\bQ$ is obtained with $\bT = \bI _p$ and $\tau = 1$. In the linear system $\bQ \bth = \bb$, $\bb$ has independent entries with uniform distribution on $(-0.5, 0.5)$.}
\label{tbl:spectrum_and_preconditioning}
\end{table}

\begin{remark}[CG preconditioners]
	There is a vast literature that studies the choice of the optimal CG preconditioner for different linear systems \citep[Ch.10]{book:saad2003}. In our analysis, we have also explored other preconditioning techniques, such as so-called incomplete Cholesky factorizations (see the supplementary material for a more detailed discussion). Nevertheless, we observed results comparable, or slightly worse, to the ones obtained with Jacobi preconditioning.
\cite{conjugate_gradient_nishimura_suchard} also proposed a prior-preconditioned version of conjugate gradient method. 
However, they focused on Bayesian sparse regression with a large number of covariates, identifying the relevant ones through shrinkage priors. Since we do not consider sparsity-inducing priors, the diagonal elements of the prior term $\bT$ are much smaller than those of the likelihood term  $\tau \X ^T \X$. This is why, in our case, using $\bT$ as a preconditioner does not improve much the spectrum of $\bQ$, leading to the same issues observed in Table \ref{tbl:spectrum_and_preconditioning}.
\end{remark}

\subsection{Outlying eigenvalues of \texorpdfstring{$\bar{\bQ}$}{barQ}}

The following theorem identifies a set of outlying eigenvalues of the matrix $\bar{\bQ}$, which always exhibits $K$ small eigenvalues and a large one near $K+1$, independently of the data design.
\begin{theorem}\label{thm:outlying_eigvals}
	Let $\bU = \X ^T \X$ be the likelihood term of $\bQ$ in \eqref{eq:crossed_posterior_prec}. Denote by $\bar{\bU}$ the matrix $Diag (\bU ) ^{-1/2} \bU Diag (\bU ) ^{-1/2}$ and let $\lambda _1 (\bar{\bU})\leq \dots \leq \lambda _p (\bar{\bU})$ be its eigenvalues. Then:
    \[\lambda _1 (\bar{\bU}) = \dots =\lambda _K (\bar{\bU}) =0\,, \hspace{2cm} \lambda_p(\bar{\bU}) = K+1\,. \] 
    The previous equalities extend to the eigenvalues $\Bar{\mu}_1 \leq \dots \leq \Bar{\mu}_p$ of $\bar{\bQ}$ as follows:
    \[ 0< \Bar{\mu}_1 \leq \dots \leq \Bar{\mu}_K \leq \frac{\max _{k = 0, \dots K} T_k}{\tau + \max _{k = 0, \dots K} T_k}\,, \qquad K+1 \leq \bar{\mu}_p \leq K+2\,.\]
\end{theorem}

The eigenvalues of Theorem \ref{thm:outlying_eigvals} corresponds to those directions where the likelihood is highly or poorly informative. This is a well-known behavior due to the additive structure of the linear predictor $\eta _i $ in \eqref{eq:crossed}. The predictor is very informative about the sum of fixed and random effects, whose direction corresponds to the eigenvalue $\bar{\mu}_p$. However, it provides little information about their relative differences, whose directions correspond to the eigenvalues $\lambda _1 (\bar{\bU}) = \dots =\lambda _K (\bar{\bU})$. 
Notice also that the upper bounds on the small eigenvalues $\Bar{\mu}_1, \dots \Bar{\mu}_K$  in Theorem \ref{thm:outlying_eigvals} are quite conservative. In all the numerical examples we have explored, the $K$ smallest eigenvalues of $\bar{\bQ}$ are very close to zero.


\subsection{Connection with graphs and adjacency matrices}
The purpose of this section is to draw a connection between the spectrum of $\bar{\bQ}$ and the one of the adjacency matrix $\bar{\bA}^\rr$ of a $K$-partite graph. We will give sufficient conditions that guarantee that the eigenvalues of $\bar{\bQ}$, except for those identified by Theorem \ref{thm:outlying_eigvals}, concentrate around $1$ as $p\to +\infty$.
The following corollary extends Theorem \ref{thm:outlying_eigvals} to the matrix $\bar{\bA} ^\rr$.
\begin{corollary}\label{thm:spec_A}
	Let $\bar{\bA}^\rr= (\bD ^\rr)^{ -1/2} \bA^\rr(\bD ^\rr)^{ -1/2} \in \mathbb{R}^{(p-1)\times (p-1)}$ be as in Section \ref{sec:notation}, and denote with $\Bar{\nu}_1 \leq \dots \leq \Bar{\nu}_{p-1}$ its eigenvalues. Then,
	\begin{equation}\label{eq:spect_A}
	\bar{\nu}_1 = \dots = \bar{\nu} _{K-1}= - (K-1)^{-1}\, ,\qquad \bar{\nu} _{p-1}=1.
\end{equation}
\end{corollary}

We now focus on the remaining part of the spectrum of $\bar{\bA}^\rr$, which lies in the interval $[\bar{\nu}_K,\, \bar{\nu}_{p-2}]$. 
If this interval becomes smaller, we would also expect the spectrum of $\bar{\bQ}$ to concentrate in a smaller interval, therefore improving the effective condition number. The following lemma quantifies this connection.

\begin{lemma}\label{lemma:interlace}
	Under Model \ref{mdl:crossed}, consider $\Bar {\bQ}\in \mathbb{R}^{p\times p}$ and $\Bar{\bA}^\rr$ as above, with eigenvalues $\Bar{\mu}_1 \leq \dots \leq \Bar{\mu}_p$ and $\Bar{\nu}_1 \leq \dots \leq \Bar{\nu}_{p-1}$ respectively. Let $\Bar{\nu}_K\leq 0 \leq \Bar{\nu} _{p-2} $, then,
\begin{equation}\label{eq:bound_CN_Q_bar}
        \kappa_{K+1,p-2}(\bar{\bQ}) \leq \dfrac{1 + (K-1) \Bar{\nu} _{p-2} }{ 1 +  (K-1) \Bar{\nu} _{K} }\,.
\end{equation}
\end{lemma}

Lemma \ref{lemma:interlace} provides an interesting connection between the effective condition number $\kappa_{K+1,p-2}(\bar{\bQ}) $, and the concentration of eigenvalues of $\bar{\bA}^\rr$ around $0$. This allows to focus on the spectrum of $\bar{\bA}^\rr$ and to leverage results from random graph theory, in order to bound the condition number of $\bar{\bQ}$ using \eqref{eq:bound_CN_Q_bar}.

\begin{remark}
    The upper-bound in \eqref{eq:bound_CN_Q_bar} is quite tight in the regimes that we consider. When the bound becomes uninformative for $\bar{\nu}_K = -1/(K-1)$, we numerically observed that $\kappa_{K+1,p-2}(\bar{\bQ})$ is very large, in the order of several hundreds.
    Theoretically, if $\bar{\nu}_K = -1/(K-1)$, we can prove that	$\kappa_{K+1,p-2}(\bar{\bQ})$ is lower bounded by $\tau\, (\max _k T_k)^{-1}$, hence the effective condition number can be arbitrarily large depending on the value of $\max _k T_k$.
    Heuristically, if we want $\kappa_{K+1,p-2}(\bar{\bQ}) $ to be small, we need $\bar{\nu}_K$ to be bounded away from $-1 / (K-1)$.
	
    However, it is important to remark that the condition $\bar{\nu}_K> -1 /(K-1)$ is only sufficient for the fast convergence of CG. Clearly, it can happen that, $\bar{\nu}_r> -1/(K-1)$, for some small $r>K$, even if $\bar{\nu }_K = -1/(K-1)$. In such case, we can extend Lemma \ref{lemma:interlace} to bound $\kappa_{r+1,p-2}(\bar{\bQ})$, and CG would still converge fast.
    In the following analysis, we will focus on $\kappa_{K+1,p-2}(\bar{\bQ}) $  for mathematical convenience, in order to obtain a sufficient condition for CG to converge fast. Though, we will also study necessary conditions, as in Theorem \ref{prop:pairwise_connected_necessary_condition}.
\end{remark}


In order for the upper-bound in \eqref{eq:bound_CN_Q_bar} to be effective, we would like $\bar{\nu}_{p-2}$ to be small, or equivalently $1-\bar{\nu}_{p-2}$ to be large; but mostly we need $\bar{\nu}_K $ to be bounded away from $-\frac{1}{K-1}$. 

The spectral gap $ 1 - \bar{\nu}_{p-2}$ is a well-studied quantity since it is directly related to the global connectivity of a graph. 
One may wonder whether global connectivity is also associated to the gap $\bar{\nu}_K  +\frac{1}{K-1}$.
In Section \ref{sec:K=2}, we will see that, in the simple case with two factors, this association is straightforward. Indeed, we will prove that a sufficient condition to bound the effective condition number is that the associated graph is made of a unique connected component. On the other hand, in Section \ref{sec:K>2}, we will show that, when considering more than two factors, a different notion of connectivity is needed.

\subsubsection{The case \texorpdfstring{$K=2$}{K=2}}\label{sec:K=2}
In this section, we consider Model \ref{mdl:crossed} with only two factors. 
In this case, $\bar{\bA}^\rr$ represents a bipartite graph and its spectrum has a very peculiar structure.
Indeed, its eigenvalues are symmetric around zero, and, in modulo, less or equal than 1 \citep{brito2022}. 
With such property, the link with connectivity of a graph expressed in the previous section  becomes straightforward. Indeed, $1 - \bar{\nu} _{p-2} = \bar{\nu}_{2} +1$, which simplifies the bound in Lemma \ref{lemma:interlace} as follows.
\begin{corollary}\label{thm:bipartite_bound}
	Consider Model \ref{mdl:crossed} with $K=2$, then
	\begin{equation}
		\kappa_{{3, p-2}}(\bar{\bQ}) \leq \dfrac{1+ \bar{\nu}_{p-2}}{1-\bar{\nu}_{p-2}}\, .
	\end{equation}
\end{corollary}

Assume now, that $\bA^\rr$ is a proper adjacency matrix with binary entries in $\{ 0, 1\}$. 
In this case, one can show that $\bar{\nu}_{p-2}<1$ if and only if the graph associated to the adjacency matrix $\bar{\bA}^\rr$ consists of a unique connected component \citep{book:godsil}. More in general, the higher is its connectivity, the larger is the gap $1-\bar{\nu}_{p-2}$ \citep[Ch. 13.5]{book:godsil}.
With this in mind, Corollary \ref{thm:bipartite_bound} states that CG converges fast, when the graph associated to $\bar{\bA}^\rr$ is well-connected.

Clearly, higher connectivity comes at the price of having a denser support for $\bQ$ and consequently a larger cost per iteration of CG. Ideally, CG would be most efficient when the CI graph is well-connected while being very sparse. There is a vast literature of random graph theory that studies graphs with this property. 
For example, the well-known \textit{Alon-Boppana bound} \citep{alon1986} gives an upper bound on the spectral gap of regular graphs. Graphs that attain this bound are called \textit{Ramanujan}. 
Random regular graphs are also known for having great connectivity property, indeed, they can get arbitrarily close to the Alon-Boppana bound, with high probability, as the number of vertices goes to infinity \citep{friedman2003}. 
\cite{brito2022} extends this result to the case of \textit{random biregular bipartite graph}, i.e.\ a bipartite graph sampled uniformly at random among all the possible bipartite graphs, where the $G_1$ vertices of the first part has constant degree $d_1$, and the remaining $G_2$ vertices has degree $d_2$. This is why one would expect CG algorithm to have a fast convergence rate under this random design. The following result formalizes this intuition.

\begin{theorem}\label{thm:bipartite_biregular}
Consider Model \ref{mdl:crossed} with $\bA ^\rr$ being the adjacency matrix of a bipartite, biregular random graph. Then, for any sequence $\epsilon_p \to 0$ as $p \to \infty$, it holds asymptotically almost surely that
\begin{equation}\label{eq:bipartite_biregular_eff_cn}
		\kappa _{3, p-2}(\bar{\bQ})  \leq \dfrac{1 +  (1/\sqrt{d_1} + 1/\sqrt{d_2} )+ \epsilon_p}{1 -  (1/\sqrt{d_1} + 1/\sqrt{d_2} )- \epsilon_p}\,.
\end{equation}
\end{theorem}

Theorem \ref{thm:bipartite_biregular} shows that, even with relatively small degrees $d_1$ and $d_2$, the effective condition number of $\bar{\bQ}$ becomes close to $1$. An analogous, but weaker result can be stated also for Erd\H{o}s-R\'enyi random bipartite graphs and can be found in the supplementary material.

\subsubsection{The case \texorpdfstring{$K>2$}{K>2}}\label{sec:K>2}
When considering more than two factors, the symmetry of the spectrum of $\bar{\bA}^\rr$ around $0$ is lost. One may wonder whether the direct link established by Corollary \ref{thm:bipartite_bound}, between the connectivity of the graph associated to $\bar{\bA}^\rr$, expressed by the spectral gap $1 - \bar{\nu}_{p-2}$, and the good conditioning of the matrix $\bar{\bQ}$ holds true also for $K>2$. Unfortunately, this is not the case. Example \ref{ex:pairwise_non_connected} describes a 3-partite graph with excellent connectivity, whose associated precision matrix  $\bar{\bQ}$ has $\bigo{p}$ eigenvalues close to zero.

\begin{example}\label{ex:pairwise_non_connected}
Consider the case with $K=3$ factors and $G_1 = G_2 = G_3=G$. We consider $\n= G^2$ observations, s.t. for each $i = 1,\dots,\n$, the first two factors appear with the level $\ceil{i/G} $; while the third factor has level $(i\ mod\ G) + 1$. An example of the resulting conditional independence graph is given in the supplementary material. With such design, the first and third factors, as well as the second and third, are fully connected; however, the first and second factor are very poorly connected. Indeed, each level $\theta _{1, g}$ of the first factor only connects to $\theta _{2,g}$ of the second one, and the sub-graph restricted to the first two factors has $G$ disconnected components. 
	
In this case, the spectral gap of $\bar{\bA}^\rr$ is $1- \bar{\nu}_{p-2} = 0.5$ and the CI graph has great connectivity, as each vertex is at most at distance $2$ from any other one. Nonetheless, $G+1$ eigenvalues of $\bar{\bA}^\rr$ are equal to the lower-bound of $-1/2$ in \eqref{eq:spect_A}. As a consequence, CG would need to remove many small eigenvalues before getting to a faster rate of convergence. For instance, if we take $G=200$, $\bT = \bI_p $ (the identity matrix) and $\tau =  1$, we have $\kappa_{G+3,p-2}(\bar{\bQ}) \approx 2$, but $\kappa_{G+2,p-2}(\bar{\bQ}) \approx 400$.

\end{example}

The intuition behind Example \ref{ex:pairwise_non_connected} is made more general in the following theorem, which shows that \textit{pairwise connectivity} between all factors is a necessary condition for a well conditioned precision matrix.

\begin{theorem}\label{prop:pairwise_connected_necessary_condition}
	Consider Model \ref{mdl:crossed}. Fix any permutation $\pi$ of the first $K$ integers. For each $\ell = 1,\dots, K-1$, let $C_\ell$ be the number of connected components of the sub-graph restricted to the factors $(\pi(\ell), \pi(\ell+1))$. Then $\bar{\bA}^\rr$ has at least $\sum _{\ell = 1}^{K-1}C_\ell$ eigenvalues equal to $-\frac{1}{K-1}$.
\end{theorem}

Theorem \ref{prop:pairwise_connected_necessary_condition} shows that, if we want $\bar{\nu}_{K}$ to be bounded away from $ -\frac{1}{K-1}$, one needs each pair of bipartite sub-graphs to have only one connected component, that is $C_\ell = 1$ for each $\ell$ and $\pi$. 
This is clearly a stronger property than global connectivity of  the graph of $\bar{\bA}^\rr$, and in several cases it may not be satisfied. 
For instance, some datasets may have one factor which is nested into another one, i.e.\ when each level of the former factor can be observed only with one level of the latter one (see Section \ref{sec:real_data} for an example). In this case, the sub-graph associated to these two factors would present as many connected components as the number of levels in the first factor, and $\bar{\bA}^\rr$ will have as many eigenvalues equal to $-1 / (K-1)$.
This issue becomes even more relevant when considering \textit{interaction terms} between factors. 
Consider two factors of sizes $G_1$  and $G_2$ respectively, the interaction term between them is given by a third factor of size $G_1 \cdot G_2$ that consider all the possible pairwise combinations among the levels of the two factors. 
One can notice that, by construction, the interaction term is nested inside both of the original factors. Hence, the corresponding $\bar{\bA}^\rr$ will have at least $G_1 + G_2$ eigenvalues equal to $-1/(K-1)$. As we will see in Section \ref{sec:real_data}, including nested factors and interaction terms adds to the spectrum of $\bar{\bQ}$ several ``problematic" eigenvalues close to zero, slowing down the convergence of CG algorithm.

\medskip
Unfortunately, neither pairwise connectivity between all factors guarantees $\bar{\nu}_{K}>  -\frac{1}{K-1}$, an example is given in the supplementary material. Indeed, a stronger notion of pairwise connectivity is needed to lower bound $\bar{\nu}_K$.
The following result describes such condition, and shows that the effective condition number $\kappa_{K+1, p-2}(\bar{\bQ})$ is bounded independently of $p$.

\begin{theorem}[Strong pairwise connectivity]\label{thm:suff_cond_strong_connnectivity}
	 Consider Model \ref{mdl:crossed}. For any pair of factors $k\neq h$, let $\bA^{(k, h)} \in \RR ^{(G_k + G_{h})\times (G_k + G_{h})}$ be the adjacency matrix of the bipartite graph restricted to the pair of factors $(k,h)$. Denote with $\bar{\bA}^{(k,h)}$ the normalized version obtained as $\bM ^{-1/2}\bA^{(k, h)}\bM ^{-1/2}$, with $\bM = Diag( \bA^{(k, h)}\cdot \mathbb{1}) $. Assume that
\begin{equation}\label{eq:suff_pairwise_cond}
	 \lambda^* = \sqrt{K-1} \cdot  \max \left \{ |\lambda| \colon \lambda \in Spect\left (\bar{\bA}^{(k,h)}\right ),\,  |\lambda | \neq 1 ,\, k \neq h\right \}  <1\ .
\end{equation}	 
Then, the effective condition number is bounded by
\begin{equation}
	\kappa_{{K+1, p-2}}(\bar{\bQ}) \leq \dfrac{1+ \lambda^*}{1-\lambda^*}\ .
\end{equation}
\end{theorem}

Condition \eqref{eq:suff_pairwise_cond} states that, if the second-largest eigenvalue of $\bar{\bA}^{(k,h)}$ is bounded away from $\frac{1}{\sqrt{K-1}}$ for all pairs $(k,h)$, then $\bar{\bQ}$ has a bounded condition number after the removal of few extreme eigenvalues. 
For example, if we consider a random $K$-partite $d$-regular graph, the condition $\frac{2}{\sqrt{d}}< \frac{1}{\sqrt{K-1}}$ i.e.\ $d > 4(K-1)$, would guarantee \eqref{eq:suff_pairwise_cond} asymptotically almost surely, as $p\to \infty$ (see Theorem \ref{thm:bipartite_biregular}).

\subsection{Numerical experiments}\label{sec:simulated}

The above bounds on the effective condition numbers of $\bar{\bQ}$ suggest that CG samplers are highly effective for large unstructured designs, which resembles random graphs. 
In this section, we verify these conclusions through numerical simulations.
Specifically, we consider the same numerical experiment as in the end of Section \ref{sec:negative_SLA} and compute the cost of
obtaining an approximate sample by solving \eqref{eq:cg_sampler} with CG, which is given by the product of the cost per iteration and the number of iterations needed for convergence. 
The former is of order $\bigo{n_{\bQ}}$, while the latter can be quantified as
\begin{equation}\label{eq:n.iters}
	\text{N. iterations} = \inf \left \{ k \geq 0: ||\bQ\bth ^k - \bb||_2< \epsilon ||\bb||_2 \right \}\,,
\end{equation}
where $\bth ^k$ is the approximate CG solution of the linear system $\bQ \bth = \bb$ at iteration $k$, and $\bQ$ is the precision matrix of the conditional distribution of $\bth$. 
Monitoring the relative norm of the residuals is a standard stopping criterion in libraries that implement the CG algorithm. In this case, it also allows for a fair comparison between problems with different sizes. In our numerical experiments, we will set the tolerance to $\epsilon = 10 ^{-8}$.

We consider the three scenarios presented in Figure \ref{fig:ErdosRenyi}. Table \ref{tbl:sim_data_summary} reports the number of iterations defined in \eqref{eq:n.iters} in the three different designs, while the plots in Figure \ref{fig:ErdosRenyi} show, in purple, the total number of floating points operations in logarithmic scale. 

In the first design, we have $\n \sim Binomial (G^2 , 20/G)$, so that $n_{\bQ}=\bigo{p}$. 
Here the number of CG iterations is basically constant with $p$ (see the first column of Table \ref{tbl:sim_data_summary}) and, as a result, the slope for CG in Figure \ref{fig:ErdosRenyi} is roughly equal to $1$. To be more precise, the small increase in the number of CG iterations is due to the fact that the number of disconnected components in Erd{\H{o}}s-R{\'e}nyi random graphs with a constant degree also increases as the graph size grows \citep{erdos1960}.
Instead, for designs (b) and (c) we have $n_\bQ =\bigo{p^{3/2}}$ and expected degree and connectivity of the graph that increase with $p$. This leads to an increasingly better effective condition number, which explains the decrease in the number of iterations observed in Table \ref{tbl:sim_data_summary} (second and third column) as well as why the slopes in Figure \ref{fig:ErdosRenyi} are below $1.5$ even though the cost per iteration is of order $\bigo{p^{3/2}}$.
\begin{table}[h!]
\centering

\begin{tabular}{|r|c|c|c|}
  \hline
  \multirow{2}{*}{$p$\hspace{0.3cm}} & \multicolumn{3}{|c|}{\textbf{N. iterations}} \\
  \cline{2-4}  & Case (a) & Case (b) & Case (c) \\
  \hline\hline
  100 & 17 & 22 & 45 \\
  435 & 19 & 20 & 48 \\
  1910 & 19 & 18 & 39 \\
  4000 & 19 & 17 & 34 \\\hline
\end{tabular}
\caption{Number of CG iterations \eqref{eq:n.iters} for different values of $p$, in scenarios (a)-(c).}
\label{tbl:sim_data_summary}
\end{table}

As these simulations illustrate, CG excels when it comes to large sparse unstructured random designs. Indeed, for $p$ in the order of $10^{3}$, CG converges in less than $40$ iterations in all three designs. Therefore, while this framework represents a worst-case scenario for sampling via Cholesky factorization, as its computation requires $\bigo{p^3}$ time, it is optimal for sampling via CG, whose number of iterations is independent of the dimension $p$.

\section{Application to Generalized Linear Mixed Models}\label{sec:numerics}

We illustrate the application of the CG sampler to the more general class of GLMMs. Relative to Model \ref{mdl:crossed}, we allow for the inclusion of multivariate fixed effects, random slopes and interaction terms, as well as non-Gaussian likelihoods with data augmentation. 

\subsection{Model and algorithms}\label{sec:glmms}
We consider GLMMs of the form
\begin{equation}\label{eq:GLMMs}
\begin{aligned}
	y_{i} \mid \eta_{i} \sim f\left(y_{i} \mid \eta_{i}\right),
	& \quad \eta_{i}=\bx_{i, 0}^{T} \bth_0+\sum_{k=1}^{K} \boldsymbol{x}_{i, k}^{T} \bth_{k},  
	&i=1, \ldots, \n, \\
	\bth_{k, g} \sim \mathcal{N}\left(\mathbf{0}, \bT^{-1}_{k}\right), 
	& \quad \bth_{k}=\left(\bth_{k, 1}^{T}, \cdots, \bth_{k, G_{k}}^{T}\right)^{T}, \quad g=1, \ldots, G_{k} ; 
	&k=1, \ldots, K.
\end{aligned}
\end{equation}
Here $\bth_0 \in \RR^{D_{0}}$ denotes the vector of fixed effects and  $\bx_{i,0} \in \RR^{D_{0}}$ denotes the corresponding covariates. 
For each factor $k$ and level $g$, we consider the random effect $\bth_{k, g} \in \RR^{D_{k}}$, for $g=1, \ldots, G_{k}$, with corresponding covariates $\boldsymbol{w}_{i, k} \in \RR^{D_{k}}$.
Up to now, we only considered random intercept models with $D_{k}=1$ and $\boldsymbol{w}_{i, k}=1$, while we now also consider \textit{random slopes} models with $D_{k}>1$. 
As before, the vector $\boldsymbol{z}_{i, k} \in\{0,1\}^{G_{k}}$ encodes the level of the factor $k$ to which observation $y_i$ is assigned to. With such notation, $\bth_{k}$ has covariates $\boldsymbol{x}_{i, k}=\boldsymbol{z}_{i, k} \otimes \boldsymbol{w}_{i, k} \in \RR^{G_{k} D_{k}}$, where $\otimes$ denotes the Kronecker product.

In previous sections, we discussed the case of Gaussian likelihood where $f\left(y_i \mid \eta_{i}\right)=\mathcal{N}\left(\eta_{i}, \tau^{-1}\right)$.
We consider also binomial likelihood with logistic link function $f(y_i| \eta_i)= \mathrm{Binomial}\left(n_i, (1+\exp(-\eta_i))^{-1}\right)$, using the Polya-Gamma augmentation \citep{polya} defined as 
\[
\begin{aligned}
p(y_i, \omega_i \mid \eta_i) = \frac{1}{2^{n_i}} \exp\left(\left[y_i - n_i/2\right]\eta_i - \omega_i \frac{\eta_i^2}{2}\right) f_{PG}(\omega_i | n_i, 0)\,,
\end{aligned}
\]
where $f_{PG}(\omega| b,c)$ indicates a Polya-Gamma variable with parameters $b$ and $c$.
We assign improper flat priors to the fixed effects parameters $\bth_0$ and we assign Wishart priors, $W(\alpha_k , \bPhi _k ^{-1})$ to the matrices $\bT_{k}$, 
which simplify to gamma distributions when $D_{k}=1$. 

Denote with $\bT$ and $\bm _0$ the prior precision and prior mean of $\bth$ respectively, with $\bkappa = (y_1 - n_1/2, \dots, y_\n - n_\n/2)^T$ the vector of ``centered" observations, and with $\bOmega = (\omega _1, \dots, \omega _\n )^T$ the vector of PG latent variables. 
With such notation, simple Bayesian computation leads to the following conditional posterior updates
\begin{align}
	p(\bth \mid \by, \X , \bOmega , \{\bT_k &\}_{k=1}^K) \sim \Nor _p \left (\bQ ^{-1} (\bT \bm _0 + \X ^T \bkappa),\ \bQ ^{-1}\right ), &\label{eq:gibbs_theta}\\ 
	p(\omega_i \mid \by, \X, \bth) & \overset{ind.}{\sim} PG (n_i,\ \x_i ^T \bth ), \qquad &\forall i=1,\dots,\n,\label{eq:gibbs_omega}\\
	p(\bT_k \mid \by, \X, \bth)  & \overset{ind.}{\sim} W \left ( \alpha _k  + G_k ,\ (\bPhi _k  + \sum _{g = 1}^{G_k} \bth _{k,g} \bth _{k,g}^T  )^{-1}\right ), &\quad \forall k = 1, \dots, K,\label{eq:gibbs_prec}
\end{align}
where $\bQ = \bT + \X ^T \bOmega \X$. The resulting Gibbs sampler is described in the supplementary material. The computational bottleneck is the $\bigo{p^3}\gg \bigo{N}$ cost for exact sampling from \eqref{eq:gibbs_theta}, which we seek to reduce with CG samplers.

\subsection{Application to voter turnout and student evaluations}\label{sec:real_data}
We explore the performances of the CG sampler when used to sample from the distribution in \eqref{eq:gibbs_theta} in high-dimensional contexts. 
We consider the following two data sets:
\begin{enumerate}
	\item A survey data of the 2004 American political elections \citep{ghitza2013, pfvi}. It collects the vote (Democrat or Republican) of 79148 American citizens, and contains information about their \textit{income} (5 levels), \textit{ethnic group} (4 levels), \textit{age} (4 levels), \textit{region} (5 levels) and \textit{state} (51 levels).
	\item Instructor Evaluations by Students at ETH \citep{Biometrika_Om_Giacomo_Roberts}. 
 A data set with 73421 observations for the following variables: \textit{students Id} (2972 levels), a factor denoting \textit{individual professors} (1128 levels), \textit{student's age} (4 levels), and the \textit{department of the lecture} (14 levels). The response is the \textit{ratings of lectures by the students} (discrete rating from \textit{poor} to \textit{very good}), which will be modelled as $y_i \sim Binom (4, \eta _i)$, according to the notation of \eqref{eq:GLMMs}.
\end{enumerate}

We monitor 
the number of iterations needed for CG to converge, defined as in \eqref{eq:n.iters}.
We also explored other stopping criteria, such as $||\bth ^k - \bth||_2< \epsilon_1 ||\bth||_2$ or $||\bth ^k - \bth||_\infty<\epsilon_2$, but decided not to include them in the analysis since they yielded similar results.
We set the relative accuracy threshold to $10^{-8}$, which in our simulations led to practically negligible error introduced by the CG solver. 
We also tested different accuracy levels (from $10^{-4}$ to $10^{-8}$), but opted for a more conservative one. Indeed, in our experiments once CG reached a higher convergence rate after the ``removal" of the extreme eigenvalues, it often reaches very low error 
with only a few additional iterations (see end of Section \ref{sec:cg}).

We investigate how the number of CG iterations changes depending on:
\begin{itemize}
\item \textbf{Design complexity}. 
We aim to study how the number of CG iterations varies when moving away from missing completely at random designs, which have proven very beneficial for the convergence of the CG. 
We also construct two simulated datasets with the same number of factors and levels as in the voter turnout and student evaluations data, but with a random design matrix simulated as described in Section \ref{sec:simulated}. We will compare the simulated datasets with the real ones, to isolate the contribution of the complexity of $G_\bQ$.
\item \textbf{Model complexity}. We consider the following variations of the model in \eqref{eq:GLMMs}.\\
 (a) \textit{Random intercept}. When we associate to each factor a one-dimensional random effect.\\
 (b)
 \textit{Random intercept with nested factors}. This scenario arises when certain levels of a given factor are divided into groups and the group label is also included as another factor. For the voter turnout, we consider the factor \textit{state}, which is nested inside \textit{region of the country}. For the student evaluations, the factor denoting the professor is nested inside \textit{department}.\\
 (c) \textit{Random slopes}. We include interactions between a continuous predictor and a factor, which ``modifies the slope" of the predictor depending on the level of the factor.\\
 (d) \textit{Interactions}. We add interaction terms to the  original factors in each data set, e.g.\ if we have two factors of sizes $G_1$  and $G_2$ respectively, the \textit{two-way interaction term} is a third factor with $G_1 G_2$ levels, which encodes the information about the pairwise interactions of the first two factors. Notice that each interaction term between two factors is, by construction, nested in both factors. We will also include three-way interactions.

\item \textbf{Sample size}. We can analyze the impact of sample size by considering, for each data set, a subsample of size $\n = 7000$, and a larger one with $\n = 70000$.
\item \textbf{Factor size}.
The student evaluations data set involves factors with thousands of levels, compared to tens of levels for the voter turnout.
We will see that factor size plays a crucial role when considering the number of CG iterations.
\end{itemize}
Table \ref{tbl:real_data_summary} summarizes the results of our numerical study. Each row shows the number of CG iterations for the different designs described above. We use the simple additive model (a) as a benchmark and consider only one additional term at a time.
E.g., if we include random slopes, we remove the nested factor. The only exception is the three-way interactions, as it would be unusual to consider them without the pairwise interactions. The last row presents results for the full model, which includes all the effects mentioned. Each entry in Table \ref{tbl:real_data_summary} reports the number of iterations for different sample sizes. For each data set, we include a column displaying the number of iterations obtained using simulated data.

\begin{table}[h!]
\centering 

\begin{tabular}{|l|c|c|c|c|}
\hline
  \multirow{2}{4cm}{\textbf{Case}}& \multicolumn{2}{|c|}{Voter Turnout} & \multicolumn{2}{|c|}{Students Evaluations} \\
  \cline{2-5}
  & \textbf{Real} & \textbf{Simulated} & \textbf{Real} & \textbf{Simulated} \\
  \hline\hline
  \multirow{2}{4cm}{Random intercepts} & 30 (68) & 25 (68) & 26 (4101) & 18 (4101) \\
   & 36 (68) & 24 (68) & 35 (4101) & 16 (4101) \\\hline
  \multirow{2}{4cm}{Nested factor} & 42 (73) & 28 (73) & 63 (4115) & 25 (4115) \\
   & 54 (73) & 27 (73) & 94 (4115) & 22 (4115) \\\hline
  \multirow{2}{4cm}{Random slopes} & 48 (127) & 38 (127) & 84 (5229) & 34 (5229) \\
   & 68 (127) & 40 (127) & 150 (5229) & 34 (5229) \\\hline
  \multirow{2}{4cm}{2 way interactions}  & 134 (763) & 99 (787) & 91 (19523) & 31 (23585) \\
   & 343 (786) & 190 (787) & 122 (89385) & 35 (97897) \\\hline
  \multirow{2}{4cm}{3 way interactions} & 166 (2952) & 111 (3720) & 98 (27588) & 32 (31650) \\
   & 442 (3569) & 258 (3723) & 128 (159633) & 36 (168028) \\\hline
  \multirow{2}{4cm}{Full} & 189 (3016) & 147 (3784) & 122 (28730) & 58 (32792) \\
   & 532 (3633) & 351 (3787) & 266 (160775) & 65 (169170) \\
   \hline
\end{tabular}
\caption{Average number of CG iterations. The average is obtained over 200 Gibbs sampler iterations after an initial burn-in. The number in parentheses represents $p$. In each entry, we report results obtained with $N=7\, 000$ (above) and $N=70\, 000$. As priors on $\bT_k$'s we considered independent Wishart distributions $\text{W}_{D_k} (1/10, \, 1/10\  \mathbf{I}_{D_k})$.}
\label{tbl:real_data_summary}
\end{table}

Based on the results in Table \ref{tbl:real_data_summary}, we draw the following conclusions:
\begin{itemize}
	\item In accordance with Theorem \ref{prop:pairwise_connected_necessary_condition}, we observe that introducing factors with poor pairwise connectivity results in an increase in CG iterations. This increase occurs when we include nested factors (e.g., \textit{region} for voter turnout and \textit{department} for student evaluations) or interaction terms,  which introduce nesting by construction. For the voter turnout case, the number of CG iterations increases significantly when multi-way interactions are added, moving from 36 to 343 iterations by simply introducing 2 way interactions. For the student evaluations data, adding interaction terms appears less problematic relative to the problem size, although including interactions between such large factors is usually uncommon.

	
	\item Table \ref{tbl:real_data_summary} shows, for the real data, an overall increase in the number of CG iterations when the sample size increases. This might seem counter-intuitive since, for larger $\n$, the CI graph becomes denser and better connected.
	However, when examining the spectrum of $\bar{\bQ}$ for the random intercept only case, we observe that, for small $\n$, the interval $[\bar{\mu}_{K+1},\ \bar{\mu}_{p-2}]$ is larger but most of the eigenvalues therein are tightly concentrated. In contrast, for larger $\n$, the interval $[\bar{\mu}_{K+1},\ \bar{\mu}_{p-2}]$ is smaller, consistently with the CI graph being better connected, but the spectrum is more diffuse.
	In other words, for smaller $\n$, the effective condition number $\kappa _{K+1 , p-2}(\bar{\bQ})$ is worse in accordance with Theorem \ref{thm:suff_cond_strong_connnectivity}. However, in this example, $\kappa _{s+1 , p-r}$ is better for larger $s,r$, resulting in an overall faster convergence rate.
	
	\item The comparison between these two data sets highlights the significant impact of factor size on the number of CG iterations. For the voter turnout data, which presents relatively small factors, CG can require several hundred iterations, even for solving relatively small problems on the order of a few thousand. Moreover, this issue is not specific to the voter turnout data, as a similarly large number of iterations is needed also in the simulated case with a missing completely at random design. 
	On the other hand, when considering the student evaluations case, where some factors have numerous levels,  CG proves to be very efficient, requiring only a few hundred iterations to solve significantly large linear systems in the order of $10^5$. This difference becomes even more significant in the simulated case, where approximately 50 iterations are sufficient for convergence, even for the largest problems.
\end{itemize}

\subsubsection{Comparison with Cholesky factorization}
In Table \ref{tbl:real_data_flops}, we compare the number of flops required to solve the linear system in \eqref{eq:cg_sampler} with CG method and with Cholesky decomposition under the same setting as in Table \ref{tbl:real_data_summary}. We report both $\Cost (\text{CG}) $, i.e.\ the number of flops required by CG solver, and the ratio $\Cost (\Chol) / \Cost (\text{CG})$. 
In accordance with the above theory and discussions, CG has a computational advantage in 
the students evaluation data, but not in case of the voters turnout data.
A comparison of the running times for CG and Cholesky decomposition for the experiments in Table \ref{tbl:real_data_summary} can be found in the supplementary material.
\begin{table}[h!]

\begin{tabular}{|l|c|c|c|c|}
\hline
  \multirow{2}{3.6cm}{\textbf{Case}}& \multicolumn{2}{|c|}{Voter Turnout} & \multicolumn{2}{|c|}{Students Evaluations} \\
  \cline{2-5}
  & \textbf{Real} & \textbf{Simulated} & \textbf{Real} & \textbf{Simulated} \\
  \hline\hline
  \multirow{2}{3.6cm}{Random intercepts} & 0.32 (8.2\,e04) & 0.40 (6.8\,e04) & 0.23 (1.8\,e06) & 96.34 (1.3\,e06) \\
   & 0.28 (9.6\,e04) & 0.41 (6.6\,e04) & 37.76 (6.8\,e06) & 179.38 (3.1\,e06) \\\hline
  \multirow{2}{3.6cm}{Nested effect} & 0.24 (1.3\,e05) & 0.49 (9.8\,e04) & 0.13 (5.1\,e06) & 60.99 (2.5\,e06) \\
   & 0.20 (1.6\,e05) & 0.51 (9.4\,e04) & 14.67 (2.2\,e07) & 88.27 (6.5\,e06) \\\hline
  \multirow{2}{3.6cm}{Random slopes} & 0.31 (3.3\,e05) & 0.32 (2.6\,e05) & 0.11 (7.8\,e06) & 150.27 (5.4\,e06) \\
   & 0.17 (4.9\,e05) & 0.28 (2.9\,e05) & 20.22 (4.1\,e07) & 369.70 (1.1\,e07) \\\hline
  \multirow{2}{3.6cm}{2 way interactions} & 0.45 (7.5\,e06) & 0.71 (7.0\,e06) & 0.38 (4.2\,e07) & 46.89 (1.9\,e07) \\
   & 0.20 (2.4\,e07) & 0.32 (1.5\,e07) & 28.00 (3.1\,e08) & 1538.58 (1.0\,e08) \\\hline
  \multirow{2}{3.6cm}{3 way interactions} & 0.65 (3.1\,e07) & 1.96 (3.4\,e07) & 0.15 (6.2\,e07) & 33.41 (2.6\,e07) \\
   & 0.39 (1.3\,e08) & 0.63 (9.2\,e07) & 14.86 (5.4\,e08) & 923.33 (1.7\,e08) \\\hline
  \multirow{2}{3.6cm}{Full} & 0.57 (4.4\,e07) & 1.49 (5.3\,e07) & 0.30 (9.8\,e07) & 68.14 (5.1\,e07) \\
   & 0.33 (1.8\,e08) & 0.48 (1.5\,e08) & 11.64 (1.3\,e09) & 723.46 (3.4\,e08) \\
   \hline
\end{tabular}
\caption{In parentheses, average number of flops needed for solving \eqref{eq:cg_sampler} with CG. Outside the parentheses, the ratio between the number of flops for Cholesky solver and the CG.}
\label{tbl:real_data_flops}
\end{table}

\subsection{Application to large-scale data}
Finally, we present an application to the \textit{MovieLens 25M Dataset} \citep{movielens}, 
which contains $25$ million ratings from $162$ thousand users on $62$ thousand movies.
We consider a random intercepts model that predicts ratings using \textit{userId} and \textit{movieId} as categorical predictors. The results are shown in Table \ref{tbl:large_data_summary}. 
In this example, CG converges in less than $28$ iterations, while Cost(Chol) is significantly higher, and the relative difference grows with $p$.
This experiment suggests that the asymptotic behavior predicted by the theory is coherent with large-scale real data analyses: the convergence rate of CG remains independent of $\n$ and $p$, while the cost of computing the Cholesky factor becomes a significant bottleneck when $p$ exceeds several tens of thousands.
Finally, the Jacobi preconditioning becomes crucial in this scenario, as it reduces the number CG iterations from more than a thousand to $28$. 
\begin{table}[h!]
\centering 

    \begin{tabular}{|c|c|c|c|c|}
\hline
  $\n$ & $p$ & CG iterations & Cost(CG) & Cost(Chol)/Cost(CG)   \\
  \hline \hline
  250\ 000 & 25\ 726 & 28 & 5.2\,e7 & 671 \\
  25\ 000\ 000 & 221\ 589 & 27 & 1.4\,e9 & 3140 \\
  \hline
\end{tabular}
\caption{Average number of CG iterations, Cost(CG) and Cost(Chol) for MovieLens dataset with random intercepts only model. The setting is the same as in Table \ref{tbl:real_data_summary}.}
\label{tbl:large_data_summary}
\end{table}
    

\section{Comparison with existing works}
GLMMs present significant computational challenges, especially in high-dimensional settings. 
Efficient computation for these models has been extensively studied in scenarios where the random effects have a nested dependence structure \citep{tan2018, tan2021, om_giac_note}. Successful strategies have also been proposed for more complex dependence structures, arising e.g.\ from time-series models \citep{salimans2012, tan2018} or spatial data analysis \citep{book:rue}. 
However, the complexity of the crossed dependence structure hinder the applicability of these methodologies for GLMMs with more than one factor.
An important line of of work has investigated the application of sparse linear algebra techniques within Gaussian process regression \citep{Vecchia_1, KL-minimization}. Though, such methodologies require access to the covariance matrix $\bQ ^{-1}$ and we were unable to extend these ideas to our framework.

It's known that standard techniques become inefficient when applied to high dimensional crossed effects models \citep{Gao2020, Biometrika_Om_Giacomo_Roberts}.
For this reason, various computational approaches have been explored to address this issue, including backfitting with centering \citep{Backfitting_for_crossed_random_effects}, method of moments \citep{Gao2020}, collapsed Gibbs sampler \citep{Biometrika_Om_Giacomo_Roberts, scalable_tim, ceriani2024}, and variational inference methods \citep{menictas2023streamlined, pfvi}. 

The CG method is widely recognized for its effectiveness in solving large, sparse linear systems, and it has been extensively applied also for sampling from high dimensional Gaussian distributions \citep{papandreou2010, simpson_2013, CG_review}.
However, there is a lack of theoretical studies analyzing the effectiveness of CG for GLMMs. 
The recent work of \cite{conjugate_gradient_nishimura_suchard} provides a detailed study of the use of CG in the context of large Bayesian sparse regression models. 
Differently from their analysis, we consider the case where the design matrix itself $\X$ is sparse. 

\section{Discussion}\label{sec:discussion}

We remark that the applicability of the methodologies presented in this article goes beyond Gibbs sampling. Standard inferential methodologies such as restricted maximum likelihood estimation, Laplace approximation and variational inference require the computation of the Cholesky factorization of a matrix as in \eqref{eq:post_prec}, and would suffer from the same limitations presented in Section \ref{sec:negative_SLA}. 
Indeed, some of these methodologies would benefit from the use of CG methods. In the context of ML estimation, one could employ the method of moments proposed in \cite{Gao2020}, to obtain consistent estimates of the variance components and then use CG to compute the GLS estimator. 
Also, in the context of partially factorized or unfactorized variational inference \citep{pfvi}, one is required to invert matrices with the same structure of \eqref{eq:post_prec}, which can be computed via CG.
In all these cases, the spectral analysis of Section \ref{sec:CG_theory} may be of interest to predict theoretical performances. 

This work raises both theoretical and methodological questions. For instance, establishing a formal proof of the cubic complexity of Cholesky factorization under certain random designs would be valuable. 
It would also be interesting to extend the spectral analysis of Section \ref{sec:CG_theory} to account for different blocking strategies in the sampling of $\bth$, though this analysis looks already daunting in the case $K=2$ \citep{menictas2023streamlined} and we leave its exploration to future works.
Finally, Theorem \ref{thm:outlying_eigvals} and \ref{prop:pairwise_connected_necessary_condition} shows how linear dependence among the columns of the design matrix $\X$ relates to small eigenvalues of $\bQ$, which, in general, significantly slows down the convergence of CG. A standard technique to address this issue involves imposing linear constraints on the vector of parameters $\bth$ \citep{book:wood2017, zanella2020multilevel}. 
Preliminary numerical experiments suggest that removing columns from the design matrix that are linear combinations of others, for instance, by omitting one level per factor (see Theorem \ref{thm:outlying_eigvals} and \ref{prop:pairwise_connected_necessary_condition}), can lead to a net improvement in the convergence rate of CG. The gains can be substantial, around $40\%$, in models where CG typically struggles, e.g.\ the Voter Turnout data, and more modest, around $5\%$, in settings where CG already performs well, e.g.\ the Student Evaluations data.
We leave a more detailed investigation to future work.

\bibliographystyle{apalike}
\bibliography{bibliography.bib}

\begin{thebibliography}{}

\bibitem[Alon, 1986]{alon1986}
Alon, N. (1986).
\newblock Eigenvalues and expanders.
\newblock {\em Combinatorica}, 6(2):83--96.

\bibitem[Amestoy et~al., 1996]{amestoy_davis_AMD}
Amestoy, P.~R., Davis, T.~A., and Duff, I.~S. (1996).
\newblock An approximate minimum degree ordering algorithm.
\newblock {\em SIAM Journal on Matrix Analysis and Applications},
  17(4):886--905.

\bibitem[Ashcroft, 2021]{spectrum_ER_bipartite_graphs}
Ashcroft, C.~J. (2021).
\newblock On the eigenvalues of {E}rd{\H{o}}s-{R}{\'e}nyi random bipartite
  graphs.

\bibitem[Bates et~al., 2025]{bates2025}
Bates, D., Alday, P.~M., and Kokandakar, A.~H. (2025).
\newblock Mixed-model log-likelihood evaluation via a blocked cholesky
  factorization.

\bibitem[Brito et~al., 2022]{brito2022}
Brito, G., Dumitriu, I., and Harris, K.~D. (2022).
\newblock Spectral gap in random bipartite biregular graphs and applications.
\newblock {\em Combinatorics, Probability and Computing}.

\bibitem[Ceriani and Zanella, 2024]{ceriani2024}
Ceriani, P.~M. and Zanella, G. (2024).
\newblock Linear-cost unbiased posterior estimates for crossed effects and
  matrix factorization models via couplings.
\newblock {\em preprint arXiv:2410.08939}.

\bibitem[Chung, 1997]{book:chung}
Chung, F.~R. (1997).
\newblock {\em Spectral graph theory}, volume~92.
\newblock American Mathematical Soc.

\bibitem[Erd{\H{o}}s et~al., 1960]{erdos1960}
Erd{\H{o}}s, P., R{\'e}nyi, A., et~al. (1960).
\newblock On the evolution of random graphs.
\newblock {\em Publ. math. inst. hung. acad. sci}, 5(1):17--60.

\bibitem[Friedman, 2003]{friedman2003}
Friedman, J. (2003).
\newblock Relative expanders or weakly relatively {R}amanujan graphs.
\newblock {\em Duke Mathematical Journal}, 118(1).

\bibitem[Gao and Owen, 2017]{GaoOwen2017}
Gao, K. and Owen, A. (2017).
\newblock Efficient moment calculations for variance components in large
  unbalanced crossed random effects models.
\newblock {\em Electronic Journal of Statistics}, 11.

\bibitem[Gao and Owen, 2020]{Gao2020}
Gao, K. and Owen, A.~B. (2020).
\newblock Estimation and inference for very large linear mixed effects models.
\newblock {\em Statistica Sinica}, 30:1741--1771.

\bibitem[Gelman and Hill, 2007]{book:gelman}
Gelman, A. and Hill, J. (2007).
\newblock {\em Data analysis using regression and multilevel/hierarchical
  models}, volume~3.
\newblock Cambridge University Press New York, New York, USA.

\bibitem[Ghitza and Gelman, 2013]{ghitza2013}
Ghitza, Y. and Gelman, A. (2013).
\newblock {Deep interactions with MRP: Election turnout and voting patterns
  among small electoral subgroups}.
\newblock {\em American Journal of Political Science}, 57(3):762--776.

\bibitem[Ghosh et~al., 2022]{Backfitting_for_crossed_random_effects}
Ghosh, S., Hastie, T., and Owen, A.~B. (2022).
\newblock {Backfitting for large scale crossed random effects regressions}.
\newblock {\em The Annals of Statistics}, 50(1):560 -- 583.

\bibitem[Godsil and Royle, 2001]{book:godsil}
Godsil, C. and Royle, G.~F. (2001).
\newblock {\em Algebraic graph theory}, volume 207.
\newblock Springer Science \& Business Media.

\bibitem[Golub and Van~Loan, 2013]{book:golub2013}
Golub, G.~H. and Van~Loan, C.~F. (2013).
\newblock {\em Matrix Computations}.
\newblock JHU Press.

\bibitem[Goplerud et~al., 2024]{pfvi}
Goplerud, M., Papaspiliopoulos, O., and Zanella, G. (2024).
\newblock Partially factorized variational inference for high-dimensional mixed
  models.
\newblock {\em Biometrika}, 112(2):asae067.

\bibitem[Harper and Konstan, 2015]{movielens}
Harper, F.~M. and Konstan, J.~A. (2015).
\newblock The movielens datasets: History and context.
\newblock {\em ACM Trans. Interact. Intell. Syst.}, 5(4).

\bibitem[Horn and Johnson, 2012]{book:horn_johnson}
Horn, R.~A. and Johnson, C.~R. (2012).
\newblock {\em Matrix Analysis}.
\newblock Cambridge University Press, 2 edition.

\bibitem[Katzfuss and Guinness, 2021]{Vecchia_1}
Katzfuss, M. and Guinness, J. (2021).
\newblock A general framework for {V}ecchia approximations of {G}aussian
  processes.
\newblock {\em Statistical Science}, 36(1):124 -- 141.

\bibitem[Lin and Mor\'{e}, 1999]{lin99}
Lin, C.-J. and Mor\'{e}, J.~J. (1999).
\newblock Incomplete {C}holesky factorizations with limited memory.
\newblock {\em SIAM Journal on Scientific Computing}, 21(1):24--45.

\bibitem[Menictas et~al., 2023]{menictas2023streamlined}
Menictas, M., Di~Credico, G., and Wand, M.~P. (2023).
\newblock {Streamlined variational inference for linear mixed models with
  crossed random effects}.
\newblock {\em Journal of Computational and Graphical Statistics},
  32(1):99--115.

\bibitem[Nishimura and Suchard, 2023]{conjugate_gradient_nishimura_suchard}
Nishimura, A. and Suchard, M.~A. (2023).
\newblock Prior-preconditioned conjugate gradient method for accelerated
  {G}ibbs sampling in “large n, large p” {B}ayesian sparse regression.
\newblock {\em Journal of the American Statistical Association},
  118(544):2468--2481.

\bibitem[Papandreou and Yuille, 2010]{papandreou2010}
Papandreou, G. and Yuille, A.~L. (2010).
\newblock {G}aussian sampling by local perturbations.
\newblock {\em Advances in Neural Information Processing Systems}, 23.

\bibitem[Papaspiliopoulos et~al., 2019]{Biometrika_Om_Giacomo_Roberts}
Papaspiliopoulos, O., Roberts, G.~O., and Zanella, G. (2019).
\newblock {Scalable inference for crossed random effects models}.
\newblock {\em Biometrika}, 107(1):25--40.

\bibitem[Papaspiliopoulos et~al., 2023]{scalable_tim}
Papaspiliopoulos, O., Stumpf-F{\'e}tizon, T., and Zanella, G. (2023).
\newblock {Scalable Bayesian computation for crossed and nested hierarchical
  models}.
\newblock {\em Electronic Journal of Statistics}, 17(2).

\bibitem[Papaspiliopoulos and Zanella, 2017]{om_giac_note}
Papaspiliopoulos, O. and Zanella, G. (2017).
\newblock {A note on MCMC for nested multilevel regression models via belief
  propagation}.

\bibitem[Parker and Fox, 2012]{parker_fox_2012}
Parker, A. and Fox, C. (2012).
\newblock Sampling {G}aussian distributions in {K}rylov spaces with conjugate
  gradients.
\newblock {\em SIAM Journal on Scientific Computing}, 34(3):B312--B334.

\bibitem[Polson et~al., 2013]{polya}
Polson, N.~G., Scott, J.~G., and Windle, J. (2013).
\newblock Bayesian inference for logistic models using {P}ólya–{G}amma
  latent variables.
\newblock {\em Journal of the American Statistical Association},
  108(504):1339--1349.

\bibitem[Rivin, 2002]{rivin2002}
Rivin, I. (2002).
\newblock {Counting cycles and finite dimensional Lp norms}.
\newblock {\em Advances in Applied Mathematics}, 29(4):647 -- 662.

\bibitem[Rue and Held, 2005]{book:rue}
Rue, H. and Held, L. (2005).
\newblock {\em {Gaussian Markov random fields: theory and applications}}.
\newblock Chapman \& Hall.

\bibitem[Rue et~al., 2009]{rue2009inla}
Rue, H., Martino, S., and Chopin, N. (2009).
\newblock {Approximate Bayesian inference for latent Gaussian models by using
  integrated nested Laplace approximations}.
\newblock {\em Journal of the royal statistical society: Series b (statistical
  methodology)}, 71(2):319--392.

\bibitem[Saad, 2003]{book:saad2003}
Saad, Y. (2003).
\newblock {\em {Iterative Methods for Sparse Linear Systems}}.
\newblock Other Titles in Applied Mathematics. SIAM, second edition.

\bibitem[Salimans and Knowles, 2012]{salimans2012}
Salimans, T. and Knowles, D. (2012).
\newblock Fixed-form variational posterior approximation through stochastic
  linear regression.
\newblock {\em Bayesian Analysis}, 8.

\bibitem[Sch\"{a}fer et~al., 2021]{KL-minimization}
Sch\"{a}fer, F., Katzfuss, M., and Owhadi, H. (2021).
\newblock {Sparse Cholesky factorization by Kullback--Leibler minimization}.
\newblock {\em SIAM Journal on Scientific Computing}.

\bibitem[Simpson et~al., 2013]{simpson_2013}
Simpson, D.~P., Turner, I.~W., Strickland, C.~M., and Pettitt, A.~N. (2013).
\newblock Scalable iterative methods for sampling from massive {G}aussian
  random vectors.
\newblock {\em arXiv:1312.1476}.

\bibitem[Tan, 2020]{tan2021}
Tan, L. S.~L. (2020).
\newblock Use of model reparametrization to improve variational {B}ayes.
\newblock {\em Journal of the Royal Statistical Society Series B: Statistical
  Methodology}, 83(1):30--57.

\bibitem[Tan and Nott, 2018]{tan2018}
Tan, L. S.~L. and Nott, D.~J. (2018).
\newblock {G}aussian variational approximation with sparse precision matrices.
\newblock {\em Statistics and Computing}, 28:259--275.

\bibitem[Trefethen and Bau, 1997]{book:trefethen}
Trefethen, L. and Bau, D. (1997).
\newblock {\em Numerical Linear Algebra}.
\newblock SIAM.

\bibitem[Van~der Sluis and Van~der Vorst, 1986]{sluis1986}
Van~der Sluis, A. and Van~der Vorst, H.~A. (1986).
\newblock The rate of convergence of conjugate gradients.
\newblock {\em Numerische Mathematik}, 48:543--560.

\bibitem[Van~der Vorst, 2003]{book:vorst2003}
Van~der Vorst, H.~A. (2003).
\newblock {\em Iterative {K}rylov methods for large linear systems}.
\newblock Cambridge University Press, Cambridge; New York.

\bibitem[Vono et~al., 2022]{CG_review}
Vono, M., Dobigeon, N., and Chainais, P. (2022).
\newblock High-dimensional {G}aussian sampling: A review and a unifying
  approach based on a stochastic proximal point algorithm.
\newblock {\em SIAM Review}, 64(1):3--56.

\bibitem[Wood, 2017]{book:wood2017}
Wood, S.~N. (2017).
\newblock {\em Generalized additive models: an introduction with {R}}.
\newblock CRC press.

\bibitem[Zanella and Roberts, 2020]{zanella2020multilevel}
Zanella, G. and Roberts, G. (2020).
\newblock {Multilevel linear models, Gibbs samplers and multigrid
  decompositions}.
\newblock {\em Bayesian Analysis}.

\end{thebibliography}

\newpage
\appendix

\section{List of Symbols}
\begin{tabular}{lll}
$n_\bQ $        & Number of non-zero entries of a matrix $\bQ$ \\[1em]
$G_\bQ$         & Conditional independence graph of $\bth \sim \Nor _p (\mathbf{0}, \bQ ^{-1})$\\[1em]
$\bth _k \in \RR ^{G_k}$    & Vector of random effects for factor $k$\\[1em]
$\bz _{i,k} \in \{0,1\}^{G_k}$ & One-hot vector\\[1em]
$\x _i = ( 1, \bz_{i,1}^T, \dots, \bz _{i,K}^T)^T \in \RR ^p$ & Design vector\\[1em]
$\X  = [\x_1 | \dots |\x_\n ] ^T \in \RR ^{\n\times p}$ & Design matrix\\[1em]
$\bU = \bV^T \bV$ & Likelihood contribution in $\bQ$\\[1em]
$\bT $ & Prior contribution in $\bQ$\\[1em]
$\bA = \bU - Diag (\bU)$ & Weighted adjacency matrix\\[1em]
$\bD = Diag (\bA \mathbb{1}_p)$ & Diagonal degree matrix of $\bA$\\[1em]
$\bQ ^\r, \bU ^\r, \bA ^\r \in \RR ^{p-1 \times p-1}$ & Matrices restricted to the random effects\\[1em]
$\bD ^\r $ & Diagonal degree matrix of $\bA ^\r$\\[1em]
$\bQ = Diag(\bQ) ^{-1/2}\bQ Diag(\bQ) ^{-1/2}$ & Jacobi preconditioned matrix\\[1em]
$\bar{\bA}^\r$ & Adjacency matrix normalized by $\bD ^\r$\\[1em]
$\lambda _ 1 (\bM ) \leq \dots \leq \lambda_p (\bM ) $ & Eigenvalues of a given matrix $\bM$\\[1em]
$\mu _1 \leq \dots \leq \mu _p$ & Eigenvalues of the precision matrix $\bQ$\\[1em]
$\bar{\mu} _1 \leq \dots \leq \bar{\mu} _p$ & Eigenvalues of $\bar{\bQ}$\\[1em]
$\bar{\nu }_1\leq \dots \bar{\nu}_{p-1}$ & Eigenvalues of $\bar{\bA}^\r$\\[1em]
$\bA^{(k, h)} \in \RR ^{(G_k + G_{h})\times (G_k + G_{h})}$ & Adjacency matrix restricted to factors $(k,h)$
\end{tabular}

\section{Proofs}\label{sec:suppl_proof}
\begin{proof}[Proof of Theorem \ref{thm:cost_L}]
The number of operations required by the recursive equation in \eqref{eq:chol-rec} is
\begin{equation}\label{eq:cost_sla_expr}
  \Cost(\mathrm{Chol}) = \sum_{m=1}^{p} 1 + n_{\bL,m} + n_{\bL,m}(1 + n_{\bL,m}) \,,
\end{equation}
see Theorem 2.2 of \url{https://www.tau.ac.il/~stoledo/Support/chapter-direct.pdf} for a step-by-step derivation of \eqref{eq:cost_sla_expr}.
In both cases, we have $\Cost(\mathrm{Chol}) = \bigo{\sum_{m=1}^{p} n_{\bL,m}^{2}}$ as stated in \eqref{eq:costChol_bounds}. Then, using Jensen's inequality and $\sum_{m=1}^{p} n_{\bL,m} = n_\bL$, we have
\begin{equation*}
  \sum_{m=1}^{p} n_{\bL,m}^{2} = p \left(p^{-1}\sum_{m=1}^{p} n_{\bL,m}^{2}\right) \geq p \left(p^{-1} \sum_{m=1}^{p} n_{\bL,m}\right)^2 = n_\bL^2/p \,,
\end{equation*}
which proves the lower bound in \eqref{eq:costChol_bounds}.

The upper bound in \eqref{eq:costChol_bounds}, instead, cannot be deduced purely from $\Cost(\mathrm{Chol}) = \bigo{\sum_{m=1}^{p} n_{\bL,m}^{2}}$, but rather it follows from a characterization of $\Cost(\mathrm{Chol})$ in terms of 3-cycles of an undirected graph associated to $\bL$. More precisely, define $G_{\bL}$ as the undirected graph with nodes $\{1, \dots, p\}$ and an edge between vertices $j>m$ if and only if the future set of $\theta_m$ does not separate it from $\theta_j$ in $G_\bQ$. For any $j, m \in \{1, \dots, p\}$ with $j \neq m$, we write $\{j,m\}\in G_{\bL}$ if there is an edge between $j$ and $m$ in $G_{\bL}$. By the arguments in Section \ref{sec:chol}, $L_{j\ell}$ is (a potential) non-zero if and only if $\{j, \ell\} \in G_{\bL}$. 
The dominating cost of the recursion \eqref{eq:chol-rec} to obtain $\bL$ is the computation of $\sum_{\ell=1}^{m-1} L_{m\ell}L_{j\ell}$ for $m = 2, \dots,p$ and $j = m + 1, \dots, p$. Ignoring multiplications by zero and summation of zeros, this corresponds to
\begin{equation}\label{eq:cost_3cycles}
  \bigo{|\{(\ell, m, j): \{m, \ell\} \in G_\bL, \; \{j, \ell\} \in G_\bL \text{ and } 1 \leq \ell < m < j \leq p\}|}
\end{equation}
operations. By definition of $G_{\bL}$, if $\{\ell, m\} \in G_{\bL}$ and $\{\ell, j\} \in G_{\bL}$ for $\ell < m < j$, then also $\{\ell, m\} \in G_{\bL}$. Thus, the cost in \eqref{eq:cost_3cycles} coincides with the number of 3-cycles in $G_\bL$. The upper bound in \eqref{eq:costChol_bounds} then follows, noting that the number of 3-cycles in an undirected graph with $n_\bL$ edges is less or equal than $n_\bL^{1.5}$, see e.g.\ \cite[Theorem 4]{rivin2002}.
\end{proof}

\begin{proof}[Proof of Proposition \ref{prop:def_ordering}]
Let $\bth\in \RR ^p$ be an arbitrary permutation of $(\theta_0, \bth _1, \dots, \bth_K)$ and denote by $\ell$ the position of $\theta_0$ in $\bth$. We show that, whenever $\ell < p$, switching the positions of $\theta_0$ and the variable immediately after it in $\bth$ does not increase $n_\bL$.
Such switching of positions does not change the future set of $\theta_m$ for any $m\notin\{\ell, \ell + 1\}$ and thus leaves also $n_{\bL,m}$ unchanged by its definition in \eqref{eq:n_Lm}. Moreover, since $\theta_0$ is connected to all other variables in $G_\bQ$, it follows that $n_{\bL,m}$ equals the maximum value of $p - m + 1$ for all $\theta_m$ located after $\theta_0$ in $\bth$ and for $\theta_0$ itself. Thus, when $\theta_0$ is in position $\ell$ both $n_{\bL,\ell}$ and $n_{\bL,\ell+1}$ take their maximal values and moving $\theta_0$ to position $\ell+1$ cannot increase the value of $(n_{\bL,\ell} + n_{\bL,\ell+1})$.

\end{proof}

\begin{proof}[Proof of Proposition \ref{prop:unfavourable}]
Recall that the CI graph of Example \ref{ex:worst_sla} was defined as follows.
Let $K=2$ and $G_1 = G_2 = G$. Fix an integer $d \geq 2$, and for each $g = 1, \dots, G-1$, connect the vertex $\theta _{1,g}$ to $\theta_{2,j}$ if at least one of the following conditions hold:
\begin{enumerate}
  \item[(a)] $g=j$
  \item[(b)] $d(g-1) \leq j-2 < dg \mod (G-1)$ and $g<j$
  \item[(c)] $d(g-1) \leq j-1 < dg \mod (G-1)$ and $g>j$
\end{enumerate}
For $g=G$ connect $\theta _{1,G}$ to all the $\theta_{2,j}$'s that have degree less or equal than $d$.

Define the function $r:\{1, \dots, G\} \mapsto \{1, \dots, G\}$ as $r(1) = 1$ and $r(j) = \lceil d^{-1}(j - 1)\rceil$ for $j = 2, \dots, G$, so that $\theta_{1,r(j)}$ is connected to $ \theta _{2,j}$ for all $j \geq 1$, since (b) holds true. Then, for every couple $j$ and $m$ such that $r(j) \leq m \leq j$ we now show that the future set of $\theta _{2,m}$ does not separate it from $\theta_{2,j}$ in $G_\bQ$. We construct a path in $G_\bQ$ between $\theta_{2,m}$ and $\theta_{2,j}$ that goes through $\theta_{2,1}$ (Figure \ref{fig:path} below illustrates how to construct it in the case $j=5$ and $m=3 \geq r(j) = r(5) = 2$). 
Since $ \theta _{2,j}$ is connected to $\theta_{1,j}$ and $\theta_{1,r(j)}$ for all $j \geq 1$ (by (a) and (b)), the path going from $\theta_{2,j}$ to $\theta_{1,r(j)}$ to $\theta_{2,r(j)}$ to $\theta_{1,r(r(j))}$ to $\theta_{2,r(r(j))}$ etc.\ is supported on $G_\bQ$. Also, since $r(\ell) < \ell$ for all $\ell \geq 2$ and $r(1) = 1$, the above path eventually reaches $\theta_{2,1}$. The same strategy can be applied to construct a path in $G_\bQ$ from $\theta_{2,m}$ to $\theta_{2,1}$. Joining the two above paths at $\theta_{2,1}$, we obtain a path from $\theta_{2,m}$ to $\theta_{2,j}$ in $G_\bQ$. The assumption $r(j) \leq m \leq j$ together with $r(\ell) \leq \ell$ for all $\ell$ ensures that such path does not involve elements in the future set of $\theta_{2,m}$ apart from $\theta_{2,j}$. Thus, by \eqref{eq:n_Lm}, $L[\theta_{2,g};\, \theta_{2,m}]$ is a potential non-zero whenever $r(j) \leq m \leq j$, meaning that the row of $\bL$ corresponding to $\theta_{2,j}$ contains at least $j - r(j) + 1$ potential non-zeros. Summing over $j$ we obtain
\begin{equation*}
  n_{\bL} \geq \sum_{j=2}^{G} (j - r(j) + 1)
    \geq \sum_{j=2}^{G} \left(j - \frac{j - 1}{d}\right)
    \geq \frac{d - 1}{d} \sum_{j=2}^{G} j
    = \frac{d - 1}{d}\left(\frac{G(G + 1)}{2} - 1\right) \,.
\end{equation*}
Thus, $n_{\bL} = \bigo{G^2}$ which also implies $\Cost(\mathrm{Chol}) = \bigo{G^3}$ by the lower bound in Theorem \ref{thm:cost_L}. The statements about $n_{\bL}$ and $\Cost(\mathrm{Chol})$ follow from the above equalities and $p = 1 + 2G$. 
\end{proof}

\begin{figure}[h!]
    \centering
\noindent \begin{minipage}{0.45\textwidth}
\begin{tikzpicture}[
    node distance=1.2cm and 3.5cm,
    every node/.style={draw, circle, minimum width=1.2cm, minimum height=0.8cm},
    line/.style={-},
    redline/.style={thick, red},
    blueline/.style={thick, blue}
]

\node (theta11) at (0, 6) {$\theta_{11}$};
\node (theta12) [below=of theta11] {$\theta_{12}$};
\node (theta13) [below=of theta12] {$\theta_{13}$};
\node (theta14) [below=of theta13] {$\theta_{14}$};
\node (theta15) [below=of theta14] {$\theta_{15}$};
\node (theta16) [below=of theta15] {$\theta_{16}$};
\node (theta17) [below=of theta16] {$\theta_{17}$};

\node (theta21) [right=of theta11] {$\theta_{21}$};
\node (theta22) [below=of theta21] {$\theta_{22}$};
\node (theta23) [below=of theta22] {$\theta_{23}$};
\node (theta24) [below=of theta23] {$\theta_{24}$};
\node (theta25) [below=of theta24] {$\theta_{25}$};
\node (theta26) [below=of theta25] {$\theta_{26}$};
\node (theta27) [below=of theta26] {$\theta_{27}$};

\draw[redline] (theta11) -- (theta21);
\draw[redline] (theta11) -- node[draw=none, midway, above, inner sep=0pt, anchor=south] {\textcolor{red}{$r(2) =1$}}  (theta22);
\draw[line] (theta11) -- (theta23);

\draw[redline] (theta12) -- (theta22);
\draw[line] (theta12) -- (theta24);
\draw[redline] (theta12) -- node[draw=none, midway, above, inner sep=0pt, anchor=south] {\textcolor{red}{$r(5 ) =2$}} (theta25);

\draw[line] (theta13) -- (theta23);
\draw[line] (theta13) -- (theta26);
\draw[line] (theta13) -- (theta27);

\draw[line] (theta14) -- (theta21);
\draw[line] (theta14) -- (theta22);
\draw[line] (theta14) -- (theta24);

\draw[line] (theta15) -- (theta23);
\draw[line] (theta15) -- (theta24);
\draw[line] (theta15) -- (theta25);

\draw[line] (theta16) -- (theta25);
\draw[line] (theta16) -- (theta26);
\draw[line] (theta16) -- (theta27);

\draw[line] (theta17) -- (theta21);
\draw[line] (theta17) -- (theta26);
\draw[line] (theta17) -- (theta27);

\end{tikzpicture}
\end{minipage}
\hspace{0.05\textwidth}
\begin{minipage}{0.45\textwidth}

\begin{tikzpicture}[
    node distance=1.2cm and 3.5cm,
    every node/.style={draw, circle, minimum width=1.2cm, minimum height=0.8cm},
    line/.style={-},
    redline/.style={thick, red},
    blueline/.style={thick, blue}
]

\node (theta11) at (0, 6) {$\theta_{11}$};
\node (theta12) [below=of theta11] {$\theta_{12}$};
\node (theta13) [below=of theta12] {$\theta_{13}$};
\node (theta14) [below=of theta13] {$\theta_{14}$};
\node (theta15) [below=of theta14] {$\theta_{15}$};
\node (theta16) [below=of theta15] {$\theta_{16}$};
\node (theta17) [below=of theta16] {$\theta_{17}$};

\node (theta21) [right=of theta11] {$\theta_{21}$};
\node (theta22) [below=of theta21] {$\theta_{22}$};
\node (theta23) [below=of theta22] {$\theta_{23}$};
\node (theta24) [below=of theta23] {$\theta_{24}$};
\node (theta25) [below=of theta24] {$\theta_{25}$};
\node (theta26) [below=of theta25] {$\theta_{26}$};
\node (theta27) [below=of theta26] {$\theta_{27}$};

\draw[blueline] (theta11) -- (theta21);
\draw[line] (theta11) -- (theta22);
\draw[blueline] (theta11) -- node[draw=none, midway, above, inner sep=0pt, anchor=south] {\textcolor{blue}{$r(3 ) =1$}} (theta23);

\draw[line] (theta12) -- (theta22);
\draw[line] (theta12) -- (theta24);
\draw[line] (theta12) -- (theta25);

\draw[line] (theta13) -- (theta23);
\draw[line] (theta13) -- (theta26);
\draw[line] (theta13) -- (theta27);

\draw[line] (theta14) -- (theta21);
\draw[line] (theta14) -- (theta22);
\draw[line] (theta14) -- (theta24);

\draw[line] (theta15) -- (theta23);
\draw[line] (theta15) -- (theta24);
\draw[line] (theta15) -- (theta25);

\draw[line] (theta16) -- (theta25);
\draw[line] (theta16) -- (theta26);
\draw[line] (theta16) -- (theta27);

\draw[line] (theta17) -- (theta21);
\draw[line] (theta17) -- (theta26);
\draw[line] (theta17) -- (theta27);

\end{tikzpicture}
\end{minipage}
    \caption{CI graph of Example \ref{ex:worst_sla} obtained for $d=2$ and $G=7$. The two panels illustrate how to obtain a path from $\theta_{25}$ to $\theta_{23}$ following the procedure explained in the proof of Proposition \ref{prop:unfavourable} above. In the notation of the proof, $j=5$ and $r(j) = 2 \leq m =3$. The figure on the left shows the path from $\theta_{25}$ to $\theta_{21}$ that goes through $\theta_{2\, r(5)}$ and $\theta_{2\, r(r(5))}$, while the figure on the right, shows the path from $\theta_{23}$ to $\theta_{21}$.}
    \label{fig:path}
\end{figure}


\begin{proof}[Proof of Theorem  \ref{thm:outlying_eigvals}]
	We start by proving the results relative to the likelihood term $\bU = \X ^T \X = \sum _{i=1}^\n \x_i \x_i^T$. 
	
 \emph{Proof of part 1.} First recall, from Section \ref{sec:crossed}, that $\x _i = ( 1, \bz_{i,1}^T, \dots, \bz _{i,K}^T)^T \in \RR ^p$, where $\bz _{i,k}$ have only one entry with unitary value and the remaining ones are zero. Consider the sets of vector $W = \{ \bw _k ; k = 1, \dots , K \}$, where each $\bw _k =  ( -1 , \mathbf{0}_{G_1}^T, .., \mathbb{1}_{G_k}^T, .. , \mathbf{0}_{G_K}^T)^T$, and the vector of ones $\mathbb{1}_{G_k}$ is in the position relative to the \textit{k}-th factor. 
  Then, for all $i=1, \dots ,\n$ and $k = 1\dots, K$, we have $\x _i ^T \bw _k = \mathbf{0}_p$. Thus, $\bU \bw _k = \mathbf{0}_p$ for each $k$. 
  Since the vectors in $W$ are linearly independent, we can deduce that $dim(Null(\bU)) \geq K$. Also, since $\bU $ is positive semidefinite, then at least the first $K$ eigenvalues must be equal to zero. The same holds for $\bar{\bU}$, since the matrix $Diag (\bU)$ is positive definite by hypothesis ($U_{ii}\geq 1$ for each $i$).
  We have thus shown $\lambda _1 (\bar{\bU}) = \dots =\lambda _K (\bar{\bU}) =0$, as desired.
	
  \emph{Proof of part 2.} 
  Recall that $\bU = Diag (\bU) + \bA$, where $\bA$ can be interpreted as a weighted adjacency matrix representing a $(K+1)$-partite graph (we can think of the global effect $\theta_0$ as an additional part with one single vertex). We define $\bar{\bA}$ as the adjacency matrix $\bA$ normalized by the row-wise sum of its elements, i.e.\ $\bar{\bA}= \bD^{-1/2}\bA \bD^{-1/2}$.
  As we showed in Section \ref{sec:CG_theory}, it holds that
	\[
		\bD _{ii}= \sum _{j=1}^p A_{ij} = K U_{ii}\,.
	\]
It follows that $\bar{\bU}$ and $\bar{\bA}$ satisfy $\bar{\bU} = \bI _p + K \bar{\bA}$ and, in particular, they share the same eigenvectors. Finally, recall that the largest eigenvalue of any normalized adjacency matrix $\bar{\bA}$ is $\lambda_p (\bar{\bA}) = 1$ \citep{book:chung}.
As a consequence, the largest eigenvalue relative to $\bar{\bU}$ is $\lambda_p (\bar{\bU}) = 1 + K \lambda_p (\bar{\bA}) = K+1$, as desired.

   \emph{Proof of the bounds on $\bar{\mu}_1$,\dots, $\bar{\mu}_p$.} 
Finally, we translate the spectral results on $\bar{\bU}$ into bound on the eigenvalues of $\bar{\bQ}$. Notice that $\bar{\bQ} = \bE + \bar{\bU}$ where $\bE$ is a diagonal matrix with elements $E_{ii}= (T_{ii}+\tau U_{ii})/T_{ii}$.  Since $U_{ii}\geq 1$, we thus can bound the eigenvalues of $\bE$ as
\[
	\lambda_p(\bE) = \max_i E_{ii} 
 \,\leq\, \dfrac{T_{ii}}{T_{ii}+\tau } \leq \dfrac{\bar{T}}{\tau + \bar{T}}\,,
\]
where $\bar{T}= \max _{k = 0, \dots K} T_k$. 
We can then obtain the inequalities of Theorem \ref{thm:outlying_eigvals} combining the above inequality with Weyl's inequality \citep[Theorem 4.3.1]{book:horn_johnson}. Indeed we have
\[
	0<\bar{\mu}_K \leq \underbrace{\lambda_K (\bar{\bU})}_{=0} + \lambda_p(\bE) \leq \dfrac{\bar{T}}{\tau + \bar{T}}\,,
\]
as well as
\[
	\bar{\mu}_p \leq \lambda_p (\bar{\bU}) + \underbrace{\lambda_p(\bE)}_{\leq 1} \leq K+2
\]
and
\[
	\bar{\mu}_p \geq \lambda_p (\bar{\bU}) + \underbrace{\lambda_1(\bE)}_{\geq 0} \geq K+1\,.
\]

\end{proof}

\begin{proof}[Proof of Corollary \ref{thm:spec_A}]
We begin by defining $\bar{\bU }^\r = Diag(\bU ^\r) ^{-1/2} \bU ^\r Diag(\bU ^\r) ^{-1/2}$. Theorem \ref{thm:outlying_eigvals} applied to such matrix guarantees that its spectrum is concentrated in the interval $[0, K]$ and, more specifically:
\begin{equation}
	\lambda_1(\bar{\bU}^\r) = \dots = \lambda_1(\bar{\bU}^\r) _{K-1}= 0;\qquad
		\lambda_1(\bar{\bU}^\r) _{p-1}=K.
\end{equation}
Recalling from Section \ref{sec:notation} that $\bD^\r = (K-1)  Diag(\bU^\r)$, we have
\[
	\bar{\bU }^\r = \bI _{p-1} + (K-1) \bar{\bA}^\r\,.
\]
This implies that for each eigenvalue $ \lambda$  of $\bar{\bU }^\r$, $(\lambda - 1)/(K-1)$ is an eigenvalue for $\bar{\bA}^\r$, from which follows \eqref{eq:spect_A}.
\end{proof}

\begin{proof}[Proof of Lemma \ref{lemma:interlace}]
Since $\bar{\bQ}^\r$ is obtained from $\bar{\bQ}$ by removing the row and column relative to $\theta_0$, we can apply \textit{Cauchy interlacing theorem} \citep[Theorem 4.3.17]{book:horn_johnson},
which implies
\begin{equation}\label{eq:interlace}
	\bar{\mu}_1\leq \lambda_1 (\bar{\bQ}^\r)\leq \bar{\mu}_2\leq \dots \bar{\mu}_{p-1}\leq \lambda_{p-1} (\bar{\bQ}^\r)\leq \bar{\mu}_p \ .
\end{equation}
Note that, with the notation introduced in Section \ref{sec:notation}, it holds that
\[
\bar{\bQ}^\r = \bI_{p-1} + (\bC ^\r)^ {1/2}\bar{\bA}^\r (\bC ^\r)^ {1/2}\,,
\]
where $\bC^\r$ is the diagonal matrix with elements
$$
C_{ii}^\r = \dfrac{\tau D_{ii}^\r}{T_{ii}^\r + \tau U_{ii}^\r}=\dfrac{\tau (K-1)U_{ii}^\r}{T_{ii}^\r + \tau U_{ii}^\r} \ .
$$
Since $U_{ii}^\r>0$ by assumption and $T_{ii}^\r\geq 0$, then also $0<C_{ii}^\r\leq (K-1)$.
Thus, $\bC^\r$ is a positive definite diagonal matrix with the largest diagonal element $\gamma = \max _i C_{ii}^\r\leq K-1$. Now, if $\bar{\nu}_{p-r}\geq0$ also $\lambda_{p-r}((\bC ^\r)^ {1/2}\bar{\bA}^\r (\bC ^\r)^ {1/2})\geq 0$, and it can be upper-bounded by $\gamma\bar{\nu}_{p-1} $.
An analogous lower bound can be obtained for all those $\bar{\nu}_s \leq 0$. 
To conclude, we have
\[\begin{aligned}
	\kappa_{s+1,p-r}(\bar{\bQ}) &= \dfrac{\bar{\mu}_{p-r}}{\bar{\mu}_{s+1}}\\
	&\leq \dfrac{\lambda_{p-r}(\bar{\bQ}^\r)}{\lambda_{s}(\bar{\bQ}^\r)}\\
	&=\dfrac{1+\lambda_{p-r}((\bC ^\r)^ {1/2}\bar{\bA}^\r (\bC ^\r)^ {1/2})}{1+\lambda_{s}((\bC ^\r)^ {1/2}\bar{\bA}^\r (\bC ^\r)^ {1/2})}\\
	&\leq \dfrac{1 + (K-1) \Bar{\nu} _{p-r} }{ 1 +  (K-1) \Bar{\nu} _{s} }.
\end{aligned}\]

\medskip
%
\end{proof}

\begin{proof}[Proof of Corollary \ref{thm:bipartite_bound}]
We just need to prove that the spectrum of $\bar{\bA}^\r$ is symmetric around zero. This together with Lemma \ref{lemma:interlace} gives the result.
Recall that the matrix $\bar{\bA}^\r$ can be represented as
\[
	\bar{\bA} ^\r = \begin{bmatrix}
	\bigzero & \bB \\
	\bB ^T & \bigzero
	\end{bmatrix}\,,
 \]
with $\bB \in \RR ^{G_1 \times G_2}$.
Consider w.l.o.g. any positive eigenvalue $\lambda >0$ of $\bar{\bA}^\r$. Let $\bx$ be one of its eigenvector and split it into $\bx = [ \bx _1 ,\ \bx _2]$, where $\bx_1 \in \RR ^{G_1}$  and $\bx_2 \in \RR ^{G_2}$. Then $\bx ^* = [\bx_1 , \ - \bx_2]$ is an eigenvalue relative to $-\lambda$. Indeed:
\[
	\bar{\bA}^\r \bx ^* = \begin{bmatrix}
	- \bB ^T \bx_2 \\ \bB \bx_1
	\end{bmatrix}
	= \begin{bmatrix}
	-\lambda \bx_1 \\ \lambda \bx_2 
	\end{bmatrix}
	= -\lambda \bx ^*.
\]
Since this operation also preserve the 
multiplicity of the eigenvalues, the result holds true.
\end{proof}

\begin{proof}[Proof of Theorem  \ref{thm:bipartite_biregular}]
By Theorem 3.2 of \cite{brito2022}, the adjacency matrix $\bA ^\r$ 
satisfies
\[
\begin{aligned}
 \lambda_{p-2}(\bA ^\r) &\leq \sqrt{d_1 - 2} + \sqrt{d_2 -1 } + \epsilon _p \, , \\
 \lambda_{2}(\bA ^\r) \geq &- (\sqrt{d_1 - 2} + \sqrt{d_2 -1 }) - \epsilon _p  \, ,
\end{aligned}
\]
asymptotically almost surely, with $\epsilon _ p \to 0 $, as $p \to \infty$. 
For a bipartite biregular graph, it holds that $\bar{\bA}^\r = \bA ^\r  / \sqrt{d_1 d_2}$, then the bound above extends to the eigenvalues $\bar{\nu}_2 $ and $\bar{\nu}_{p-2}$ of $\bar{\bA}^\r$ as follows
\[
\begin{aligned}
 \bar{\nu}_{p-2} &\leq \dfrac{1}{\sqrt{d_1}} + \dfrac{1}{\sqrt{d_2}} + \epsilon _p \,,\\
 \bar{\nu}_2 \geq &- \dfrac{1}{\sqrt{d_1}} - \dfrac{1}{\sqrt{d_2}} - \epsilon _p \,.
\end{aligned}
\]
Finally, we apply Lemma \ref{lemma:interlace} with $s = q = 2$ so that, when $K=2$, equation \eqref{eq:bound_CN_Q_bar} becomes
\[
	\kappa _{3, p-2} \leq \dfrac{1 + \bar{\nu }_{p-2}}{1 + \bar{\nu }_2}\,.
\]
Combining the above bounds we obtain the desired statement.
\end{proof}

\begin{proof}[Proof of Theorem \ref{prop:pairwise_connected_necessary_condition}]
Notice that
\[
    \bar{\bU}^\r = \bI _{p-1} + (K-1) \cdot \bar{\bA}^\r\,.
\]
Then, the eigenspace relative to $-1 /(K-1)$ of $\bar{\bA}^\r$ coincides with the null space of  $\bar{\bU}^\r$. For simplicity, we will refer to the latter. 
	Notice also, that $dim(Null(\bar{\bU}^\r)) = dim(Null(\bU^\r))$, since $\bar{\bU}^\r$ is the product of $\bU^\r$ and full-rank matrices.
	In the remaining part of the proof, we will find $\sum _{\ell =1}^{K-1}C_\ell$ linearly independent eigenvalues in the null space of $\bU^\r$, hence proving Proposition \ref{prop:pairwise_connected_necessary_condition}.
 We split the proof in two parts.
	
	\textit{1. Construction of $\{\bx ^{(\ell, m)} ,\ \ell = 1, \dots, K-1, \ m = 1, \dots, C_\ell\} \subset Null(\bU^\r)$}. Consider w.l.o.g. the trivial permutation $\pi$ equal to the identity. 
    Recall that $\bU ^{(r)} = \bV ^{(r),T }\bV ^{(r)} $, where $\bV ^{(r)}\in \RR ^{\n \times p}$, with $i$-th row given by $[\bz _{i,1}^T, \dots \bz _{i,K}^T]$ and $\bz_{i,k}$ are one-hot vectors.
    For each observation $i$ and factor $k$, denote with $g_k[i]$ the position in the $i$-th row of $\bV ^{(r)}$ of the unique non-null entry of $\bz _{i,k}$. Then, for each $\bx\in \RR ^{p-1}$, the $i$-th entry of $\bV ^\r \bx$ is $ \sum _{k=1}^K x_{g_k[i]}$. With such notation, we can write
\[
	\bx ^T \bU^\r \bx = || \X ^\r \bx ||^2 = \sum _{i=1}^n \left (\sum _{k=1}^K x _{g_k [i]} \right )^2=0
\]
happens if and only if $\sum _{k=1}^K x _{g_k [i]}  =0$ for each $i=1,\dots,\n$. Hence, we have the equivalence
\begin{equation}\label{eq:kernel_U}
	\bx \in Null(\bU^\r) \qquad \Leftrightarrow \qquad \sum _{k=1}^K x _{g_k [i]}  =0, \ \forall i =1, \dots, \n\,.
\end{equation}
For any $\ell = 1, \dots, K-1$, denote with $P^{(\ell, 1)}, \dots , P^{(\ell, C_\ell)}$ and $Q^{(\ell, 1)}, \dots, Q^{(\ell, C_\ell)}$ the disconnected components in the bipartite graph restricted to the pair $(\ell,\ell+1)$. In particular, $\{P^{(\ell, m)}\}_{m=1}^{C_\ell}$ are disjoint subsets of the levels in factor $\ell$, and $\{Q^{(\ell, m)}\}_{m=1}^{C_\ell}$ are disjoint subsets of the levels in factor $\ell+1$. 
For any $\ell=1, \dots, K-1$ and $m=1,\dots, C_\ell$, we define the vector $\bx ^{(\ell, m)} \in Null(\bU^\r)$ as
\[\begin{cases}
	x^{(\ell, m)}_j = 1, & \text{if j belongs to }P^{(\ell, m)}\\
	x^{(\ell, m)}_j= -1, & \text{if j belongs to }Q^{(\ell, m)}\\
	x^{(\ell, m)}_j= 0, & \text{otherwise}\,.
\end{cases}\]
The resulting vector will look like
\[
    \bx ^{(\ell, m)}= [ 0, \dots, 0, \overbrace{0,\dots, 0,\underbrace{1,\dots, 1,}_{P^{(\ell, m)}\text{ entries}}0,\dots, 0}^{\text{Factor }\ell \text{ entries}} ,\overbrace{0,\dots, 0,\underbrace{-1,\dots, -1, \ }_{Q^{(\ell, m)}\text{ entries}}0,\dots, 0}^{\text{Factor }\ell +1 \text{ entries}} ,0,\dots, 0]^T\,.
\]
Note that $\bx^{(\ell, m)}$ satisfies \eqref{eq:kernel_U} because, if an observation $i$ does not involve $P^{(\ell, m)}$, then $x^{(\ell, m)} _{g_k[i]} = 0$ for each $k$. On the other hand, if $x_{ g_{\ell}[i]}^{(\ell, m)} = 1$, then $x_{g_{\ell+1}[i]}^{(\ell, m)}= -1$ (by definition of connected components, an edge from $P^{(\ell, m)}$ can only connect to $Q^{(\ell, m)}$), while $x^{(\ell, m)}_{ g_k[i]} = 0$ for each $k\neq \ell, \ell+1$, by construction.

\textit{2. $\{ \bx ^{(\ell, m)}\} _{\ell, m}$ are linearly independent.} Consider the linear combination
\[\sum _{\ell =1}^{K-1}\sum _{m=1}^{C_\ell}\beta ^{(\ell, m)} \bx ^{(\ell, m)} = \mathbf{0}.\] 
If we consider the first $G_1$ components (i.e. we restrict to the levels of the first factor), only $\{ \bx^{(1, m)}\} _{m=1}^{C_1}$ have non-zero entries on such components. Moreover, since they are supported on the disjoint sets $(P^{(1, m)})_{m=1}^{C_1}$, it implies that $\beta  ^{(1,1)} = \dots = \beta ^{(1,C_1)}=0$. We now consider the $G_2$ components relative to the second factor. If we don't consider the previous eigenvectors, only $\{ \bx ^{(2, m)}\} _{m=1}^{C_2}$ have non-zero entries on such components. Analogously to the previous case, we can conclude $\beta^{(2, 1)}= \dots = \beta ^{(2, C_2)}=0$. We can further extend this procedure, proving that $\beta ^{(\ell, m)}=0$ for each $\ell, m$.
\end{proof}

\begin{proof}[Proof of Theorem \ref{thm:suff_cond_strong_connnectivity}]
	Consider the matrix $\bar{\bA }^\r\in \RR ^{p-1}$, Theorem \ref{thm:outlying_eigvals} applied to it implies that $\bar{\nu}_1 = \dots \bar{\nu}_{K-1} = - 1/(K-1)$ and $\bar{\nu}_{p-1}=1$. In particular, the invariant subspace associated to these eigenvalues is the $K$-dimensional space $W=\{ \bx \in\RR ^{p-1}\colon \bx \text{ is constant factor-wise}\} = Span(\bw_1 , \dots, \bw_K) $, where $\bw _k =  ( \mathbf{0}_{G_1}^T, .., \mathbb{1}_{G_k}^T, .. , \mathbf{0}_{G_K}^T)^T$. We denote its orthogonal with $\delta W = W ^{\perp} $.
With this notation, Theorem \ref{thm:suff_cond_strong_connnectivity} becomes a straightforward consequence of Lemma \ref{lemma:interlace} and \ref{ref:lemma_lambda_star} below.
\end{proof}

\begin{lemma}\label{ref:lemma_lambda_star}
	Under assumption \eqref{eq:suff_pairwise_cond}, it holds that:
	\begin{equation}
		\max (|\bar{\nu}_{K}|, \bar{\nu}_{p-2}) = \max _{\bx \in \delta W} \dfrac{||\bar{\bA}^\r\bx||}{||\bx||}< \dfrac{\lambda^* }{K-1}\,,
	\end{equation}
	where $||\bx||$ denotes the standard Euclidean norm.
\end{lemma}
\begin{proof}
	Denote for simplicity $\eta = \frac{\lambda^*}{\sqrt{K-1}}$. Denote instead with $\bar{\bA }^\r [k,h] \in \RR ^{G_k\times G_h}$, the block of $\bar{\bA}^\r$ relative to the pair of factors $(k,h)$. It holds that:
	\[
		\bar{\bA}^{(k,h)} = (K-1)\cdot \begin{pmatrix}
			\bigzero & \bar{\bA }^\r [k,h]\\
			\bar{\bA }^\r [h,k] & \bigzero
		\end{pmatrix}\,,
  \]
	where the multiplicative factor $K-1$ comes from the fact that $\bar{\bA}^\r$ is made of $K$ cliques, hence each edge in the clique is connected to $K-1$ levels in each other factor; however when restricting to the pair $(k,h)$, it's connected to only one level. Moreover, by properties of bipartite adjacency matrix \citep{brito2022}, it holds that
	\[
		\max _{\lambda \in \sigma\left (\bar{\bA}^{(k,h)}\right ) : |\lambda | \neq 1} |\lambda| = (K-1)	\max _{\by \in \RR ^{G_h} :\ \by ^T \mathbb{1}_{G_h} = 0} \dfrac{||\bar{\bA }^\r [k,h]\by||}{||\by||}\,.
	\]
	Hence, \eqref{eq:suff_pairwise_cond} is equivalent to
	\[
			\max _{\by \in \RR ^{G_h} :\ \by ^T \mathbb{1}_{G_h} = 0} \dfrac{||\bar{\bA }^\r [k,h]\by||}{||\by||} < \frac{\eta }{K-1}\,.
	\]	
	For a given vector $\bx \in \delta W$, denote with $\bx [k]\in \RR ^{G_k}$ the vector restricted to the factor $k$; by definition of $\delta W$, it holds that $\bx [k] ^T \mathbb{1}_{G_k} =0$. Finally, taking $\bx \in \delta W$, we have
	\[\begin{aligned}
		||\bar{\bA}^\r \bx ||^2 & = \sum _{k=1}^K \left\lVert\sum _{h\neq k} \bar{\bA}^\r[k, h] \bx [h] \right\rVert^2\\
	&  	\leq \sum _{k=1}^K \sum _{h\neq k}\left\lVert \bar{\bA}^\r[k, h] \bx [h] \right\rVert^2\\
	&  	\leq  \dfrac{\eta ^2}{(K-1)^2}\underbrace{\sum _{k=1}^K \sum _{h\neq k} \left\lVert  \bx [h] \right\rVert^2}\\
	&= \dfrac{\eta ^2}{(K-1)^2} (K-1 ) ||\bx||^2\\
	&= \dfrac{\eta ^2}{K-1} ||\bx||^2\\
	&= \dfrac{(\lambda^{*})^2||\bx||^2}{(K-1)^2}\,.
	\end{aligned}\]
\end{proof}

%
%

\section{Supplementary material for Section \ref{sec:CG_theory}}\label{sec:suppl_sec_cg}
\subsection{Including multivariate fixed effects}\label{sec:suppl_mult_fixed}
Lemma \ref{lemma:interlace} provides an upper bound on $\kappa_{K+1, p-2}(\bar{\bQ})$, by interlacing the spectrum of $\bar{\bQ}$ with the one of $\bar{\bQ}^\r$ (see proof of the lemma). A similar result can also be obtained for $\bth _0 \in \RR ^{D_0}$, with $D_0 >1$, i.e.\ when we include multivariate fixed effects. In which case, the upper bound would become
\[
\kappa_{K + D_0 , p - D_0 - 1}(\bar{\bQ}) \leq \dfrac{1 + (K-1) \bar{\nu}_{p - D_0 - 1}}{1 + (K-1) \bar{\nu}_{K}}\,,
\]
where $\bar{\nu}_{K}$ and $\bar{\nu}_{p - D_0 - 1}$ would be respectively the $K$-th smallest and second-largest eigenvalues of $\bar{\bA}^\r \in \RR ^{(p-D_0)\times (p-D_0)}$. Provided that $\bth_0$ is low-dimensional, the above inequality would give an upper bound on the effective condition number after the removal of few extreme eigenvalues.
\subsection{Conjugate Gradient with Incomplete Cholesky preconditioner}
In this section, we explore alternative preconditioning techniques to lower the number of CG iterations. In particular, we will refer to the Incomplete Cholesky (IC) preconditioner \citep[Ch. 10]{book:saad2003}. IC algorithm computes
an approximate factorization $\bL \bL^T \approx \bQ$, by restraining the support of $\bL$ to a subset $S$, which can be chosen a priori or on the run (see Algorithm \ref{alg:incomplete_cholesky}).
A standard choice for $S$ is the lower-triangular support of $\bQ$ itself. Algorithm \ref{alg:incomplete_cholesky} with such choice of $S$ is also referred to as \textit{zero fill-in incomplete Cholesky}, which is a widely used CG preconditioner.

An important limitation of Algorithm \ref{alg:incomplete_cholesky} is that it may fail for general positive definite matrices. Indeed, by restraining the support of $\bL$, it may happen that $Q_{jj} - \sum_{k\colon (j,k)\in S} L_{jk}^2<0$. Several modifications have been proposed to prevent this issue, usually by scaling or shifting $\bQ$. Nevertheless, in the following numerical examples, we obtained the best results (in terms of rate of convergence of CG) by taking the absolute value of the previous quantity, hence forcing the continuation of Algorithm \ref{alg:incomplete_cholesky}.

A lot of work has been done to refine the choice of $S$ in order to improve the quality of the approximation, while preserving the computational efficiency of Algorithm \ref{alg:incomplete_cholesky}. For example, \cite{lin99} proposed a variation of the IC algorithm that selects $S$ depending on the numerical entries of the matrix and satisfies $|S| \leq p \cdot n_\bQ$. Here, the choice of $p$ allows controlling the density of the factorization, obtaining a more accurate solution at a price of additional memory and computational time. The results obtained with such preconditioner were comparable to the ones obtained with the standard IC preconditioner, hence we did not include them in this article.

Table \ref{tbl:jacobi_vs_ichol} reports the comparison between the Jacobi (see Section \ref{sec:jacobi}) and the Incomplete Cholesky preconditioner. For both the American political elections survey data and the Instructor Evaluations data set, we display the average number of CG iterations (after an initial burn-in), obtained with Jacobi and IC preconditioner, respectively. While the IC preconditioned CG requires a smaller number of iterations, such reduction is not significant (at most a factor of $4$) but comes at the additional price, in terms of time and memory complexity, of computing and inverting $\bL$. Finally, the overall cost of the two algorithms is basically equivalent in the case we have considered. However, the Jacobi preconditioning is easier to implement and does not incur in the problems of Algorithm \ref{alg:incomplete_cholesky}.

\begin{table}[h!]
\centering
\begin{tabular}{|l||c|c||c|c|}
\hline
    \multirow{2}{4cm}{\textbf{Case}}& \multicolumn{2}{|c||}{Voter Turnout} & \multicolumn{2}{|c|}{Students Evaluations} \\
  \cline{2-5}
  & \textbf{Jacobi} & \textbf{IChol} & \textbf{Jacobi} & \textbf{IChol} \\
  \hline\hline
  Random intercepts & 36 & 19 & 35 & 16 \\\hline
  Nested effect & 53 & 31 & 94 & 37 \\\hline
  Random slopes & 70 & 27 & 150 & 17 \\\hline
  2 way interactions & 338 & 104 & 121 & 53 \\\hline
  3 way interactions & 445 & 108 & 126 & 54 \\\hline
  Everything & 517 & 140 & 262 & 67 \\\hline
\end{tabular}
\caption{Average number of CG iterations with Jacobi and Incomplete Cholesky preconditioning for the data sets considered in Section \ref{sec:numerics}. We considered $N=70\, 000$ observations.}
\label{tbl:jacobi_vs_ichol}
\end{table}

\subsection{Erd\H{o}s-R\'enyi random bipartite graph}\label{sec:bipartite_ER}

In this section, we extend Theorem \ref{thm:bipartite_biregular} to the with $K=2$ factors case and design given by  an Erd\H{o}s-R\'enyi bipartite random graph (i.e.\ each edge is missing independently at random with probability $1-\pi$). The resulting bound on $\kappa _{3, p-2}(\bar{\bQ})$ is very similar to the one of Theorem \ref{thm:bipartite_biregular}, however, in this case, we need an asymptotic lower bound on the expected degree $G_1\pi$ (and $G_2\pi$) to guarantee sufficient connectivity in the limit.
\begin{theorem}\label{thm:bipartite_ER}
	Let $\bQ = Diag (\bQ) + \tau \bA$ be the posterior precision matrix in \eqref{eq:crossed_posterior_prec}.
	Let $\bA ^\r$ be the adjacency matrix of an Erd\H{o}s-R\'enyi  bipartite random graph. Let $G_2 \geq G_1 \geq 1$ and $\pi = \pi (G_1)$ be such that
	\[
		G_1 \pi = \Omega (\sqrt{G_1\pi}	\log ^3 (G_1)), \quad G_2 \pi = \Omega (\sqrt{G_2\pi}	\log ^3 (G_2))\,.
	\]
	Then, with probability tending to $1$ as $G_1$ tends to infinity
	\[
		\kappa _{3, p-2}(\bar{\bQ}) =\dfrac{\bar{\mu} _{p-2}}{ \bar{\mu}_3} \leq \dfrac{1 + \epsilon}{1 - \epsilon}\,,	
	\]
	with
	$$\epsilon= 2[2 + o(1)] \left (\dfrac{1}{\sqrt{G_1 \pi}} + \dfrac{1}{\sqrt{G_2 \pi} }+ \sqrt{\dfrac{1}{G_1\pi}+\dfrac{1}{G_2\pi}}\right ), $$
	as $G_1 \to \infty $.
\end{theorem}
\begin{proof}
	The proof of this result follows the same steps of the proof of Theorem \ref{thm:bipartite_biregular}, with the additional result from \cite{spectrum_ER_bipartite_graphs}.
\end{proof}

\subsection{Pairwise connected components is not sufficient}
\begin{example}\label{ex:count_ex_pairwise_connectivity}
	Consider a simple crossed random effects model with $K=3$ factors with size $G_1 = G_2 = G_3 = 2$. Let $N=3 $, and let the observed levels be
\[
\begin{aligned}
	&g_1[1] = 1, &g_2[1] = 1, && g_3[1] = 2 ,\\
	&g_1 [2] = 2, & g_2[2] = 2, && g_3 [2] = 2, \\
	&g_1[3] = 2, & g_2[3] = 1, && g_3[3] = 1.
\end{aligned}			
\]
The corresponding $\bU^\r$ will be
\[\bU ^\r = \begin{bmatrix}
	1 & 0 & 1 & 0 & 0 & 1\\
	0 & 2 & 1 & 1 & 1 & 1\\
	1 & 1 & 2 & 0 & 1 & 1\\
	0 & 1 & 0 & 1 & 0 & 1\\
	1 & 1 & 1 & 1 & 0 & 2
\end{bmatrix}\]

Notice that $dim(Null(\bU^\r)) = 3$ and that the following vectors generate the null-space of $\bU^\r$:
\[
\bx _1  = \begin{bmatrix}
1\\ 1 \\ - 1 \\ -1 \\ 0 \\ 0
\end{bmatrix}, \qquad
\bx _2  = \begin{bmatrix}
1\\ 1 \\ 0 \\ 0 \\ -1 \\ -1
\end{bmatrix}, \qquad
\bx _3  = \begin{bmatrix}
1\\ 0 \\ - 1/2 \\ 1/2 \\ 1/2 \\ -1/2
\end{bmatrix}\,.
\]
In particular, the first two vectors are the directions identified by Theorem \ref{thm:outlying_eigvals}, while the third one is specific to the symmetry of this particular example. By construction also $\bar{\bU} ^\r$ has a null-space of dimension $3$. Finally, because of the relationship in \eqref{eq:spect_A}, the corresponding eigenvalues of $\bar{\bA}^\r $ will satisfy $\bar{\nu}_1 = \bar{\nu}_2 = \bar{\nu}_3 = -1/2$.

To conclude, this example provides a design where the bipartite subgraph of any pair of factors is made of a unique connected component, but $\bar{\nu}_K = - \frac{1}{K-1}$. 
\end{example}

\section{Running times of experiments in Section \ref{sec:numerics}}

We compare the running times of the Cholesky solver and the Conjugate Gradient solver for solving linear systems involving symmetric positive definite matrices. The results are obtained using a laptop equipped with a 12th Gen Intel(R) Core(TM) i7-1255U @ 1.70 GHz CPU and 32 GB of RAM.
The solvers are implemented in Julia. The Cholesky decomposition is obtained via the built-in \texttt{cholesky} function from the \texttt{LinearAlgebra} module, which uses approximate minimum degree (AMD) ordering by default \citep{amestoy_davis_AMD}.
The CG solver uses the \texttt{cg} function from the \texttt{IterativeSolvers.jl} package.

We remark that several implementation-level detail influence the observed running times. E.g.\ multi-dimensional arrays in Julia are stored in column-major order, which favors solvers like Cholesky, and can be suboptimal for iterative solvers like CG, which involve repeated sparse matrix-vector multiplications.

The results would differ substantially if the same experiment were replicated in Python, using \texttt{SciPy} sparse matrix routines.
For example, the Cholesky decomposition for the student evaluation data with interactions requires several minutes to be computed, while CG takes less than a second to run until convergence.

\begin{table}[h!]

\begin{tabular}{|l|c|c|c|c|}
\hline
  \multirow{2}{3.4cm}{\textbf{Case}}& \multicolumn{2}{|c|}{Voter Turnout} & \multicolumn{2}{|c|}{Students Evaluations} \\
  \cline{2-5}
  & \textbf{Real} & \textbf{Simulated} & \textbf{Real} & \textbf{Simulated} \\
  \hline\hline
  \multirow{2}{3.4cm}{Random intercepts} & 3.14 (0.05 ms) & 3.80 (0.04 ms) & 3.64 (1.17 ms) & 16.77 (1.23 ms) \\
   & 2.67 (0.23 ms) & 3.32 (0.17 ms) & 8.62 (15.30 ms) & 17.93 (7.36 ms) \\\hline
  \multirow{2}{3.4cm}{Nested effect} & 2.99 (0.25 ms) & 5.06 (0.15 ms) & 2.02 (7.21 ms) & 32.42 (2.10 ms) \\
   & 2.03 (0.31 ms) & 3.12 (0.21 ms) & 3.11 (41.58 ms) & 11.53 (12.93 ms) \\\hline
  \multirow{2}{3.4cm}{Random slopes} & 2.27 (0.54 ms) & 2.89 (0.43 ms) & 1.22 (3.62 ms) & 12.15 (11.56 ms) \\
   & 1.36 (0.76 ms) & 2.30 (0.49 ms) & 2.33 (79.07 ms) & 14.37 (21.92 ms) \\\hline
  \multirow{2}{3.4cm}{2 way interactions} & 1.10 (10.25 ms) & 0.89 (9.31 ms) & 0.84 (0.11 s) & 3.40 (55.93 ms) \\
   & 0.33 (30.46 ms) & 0.57 (18.32 ms) & 1.71 (0.82 s) & 19.79 (0.31 s) \\\hline
  \multirow{2}{3.4cm}{3 way interactions} & 0.91 (58.46 ms) & 1.43 (60.94 ms) & 0.78 (0.14 s) & 2.53 (75.93 ms) \\
   & 0.72 (80.20 ms) & 0.52 (0.14 s) & 1.23 (1.34 s) & 12.88 (0.49 s) \\\hline
  \multirow{2}{3.4cm}{Full} & 0.74 (76.03 ms) & 1.04 (92.88 ms) & 0.45 (0.28 s) & 2.61 (0.14 s) \\
   & 0.52 (0.11 s) & 0.75 (98.51 ms) & 0.61 (3.71 s) & 9.80 (0.95 s) \\
   \hline
\end{tabular}
\caption{In parentheses, average time needed for solving \eqref{eq:cg_sampler} with CG. Outside the parentheses, we report the ratio between the average time of Cholesky solver and CG solver. We consider the same setting as in Table \ref{tbl:real_data_summary} of the main text.}
\end{table}

\begin{table}[h!]
\centering 

    \begin{tabular}{|c|c|c|c|c|}
\hline
  $\n$ & $p$ & Time(CG) & Time(Chol) & Time(Chol)/Time(CG)   \\
  \hline \hline
  250\ 000 & 25\ 726 & 0.31 s & 3.38 s & 11.0 \\
  25\ 000\ 000 & 221\ 589 & 5.85 s & 2 m 43 s & 28.0 \\
  \hline
\end{tabular}
\caption{Average number of CG iterations for MovieLens dataset with random intercepts only model. The setting is the same as in Table \ref{tbl:large_data_summary} of the main text.}
\end{table}

\section{Pseudocode of the algorithms}\label{sec:suppl_pseudocode}
\subsection{Cholesky factorization}\label{sec:cholesky_supplement}
Algorithm \ref{alg:cholesky} shows a pseudocode of the algorithm that allows computing the Cholesky factorization of a given positive-definite matrix $\bQ\in \RR ^{p\times p}$. A simple modification of Algorithm \ref{alg:cholesky} allows computing an approximate factorization $\hat{\bL}$, by restricting the support of $\bL$ to a given subset $S \subseteq \{ (i, j) \colon i \geq j,\ i = 1,\dots, p,\ j = 1, \dots, p\}$. The resulting algorithm is called incomplete Cholesky factorization \citep{book:golub2013} and it is outlined in Algorithm \ref{alg:incomplete_cholesky}.

\begin{minipage}{0.48\textwidth}
\begin{algorithm}[H]
\caption{Cholesky factorization.} \label{alg:cholesky}
			\textbf{Input: }$\bQ \in \RR^{p \times p}$ symmetric and positive definite\\
			\textbf{Output:} $\bL\in \RR^{p \times p}$ lower triangular\ s.t.\ $ \bQ = \bL \bL^T$\\
			\begin{algorithmic}[1]
				\For{$j \in \{ 1, \dots, p \}$}
				\State $L_{jj} \gets \sqrt{Q_{jj} - \sum_{k\colon (j,k)\in S} L_{jk}^2}$ 
					\For{$i \in \{ j + 1, \dots, N \}$ }
						\State $L_{ij}\gets Q_{ij}/L_{jj}$
						\For{$ k \in \{ 1, \dots, j-1 \}$ }
							\State $L_{ij} \gets L _{ij}- \left (L_{ik}\cdot L_{jk}/L_{jj}\right )$
						\EndFor
					\EndFor
				\EndFor \\
				\Return $L$\;
	\end{algorithmic}
\end{algorithm}
\end{minipage}
\begin{minipage}{0.48\textwidth}
\begin{algorithm}[H]
\caption{Incomplete Cholesky} \label{alg:incomplete_cholesky}
			\textbf{Input: }$\bQ \in \RR^{p \times p}$ symmetric, \textcolor{red}{S} sparsity set\\
			\textbf{Output:} $\bL\in \RR^{p \times p}$ lower triangular\ with $ Supp(\bL) \subset S$\\
			\begin{algorithmic}[1]
				\For{$j \in \{ 1, \dots, p \}$}
				\State $L_{jj} \gets \sqrt{Q_{jj} - \sum_{k\colon (j,k)\in S} L_{jk}^2}$ 
					\For{$i >j$ \textcolor{red}{: $(i, j) \in S$}}
						\State $L_{ij}\gets Q_{ij}/L_{jj}$
						\For{$ k < j$ \textcolor{red}{: $( i, k ), ( j, k ) \in S$}}
							\State $L_{ij} \gets L _{ij}- \left (L_{ik}\cdot L_{jk}/L_{jj}\right )$
						\EndFor
					\EndFor
				\EndFor \\
				\Return $L$\;
	\end{algorithmic}
\end{algorithm}
\end{minipage}
\subsection{Sampling with Cholesky factorization}
For a given Cholesky factor $\bL$, sampling from a Gaussian distribution with precision matrix $\bQ  = \bL \bL ^T$ is straightforward. E.g.\ one can obtain a sample $\bth \sim \Nor (\bQ ^{-1} \bm, \bQ ^{-1})$, by first sampling $\bz \sim \Nor (\mathbf{0},\bI _p)$, and  then solving the following two linear system
\[
    \bL \bw = \bm\,,\qquad 
    \bL ^T \bth = \bw + \bz\,.
\]
Specifically, one needs to solve a lower and an upper-triangular linear system respectively, which can be done with $\bigo{n_\bL}$ cost via forward and backward substitution.

\subsection{Conjugate gradient}
For the conjugate gradient algorithm, we refer to the formulation of \cite{book:saad2003}. Namely, we refer to Algorithm 6.17 and to Algorithm 9.1, for the preconditioned version. For simplicity, we report them in Algorithm \ref{alg:cg_alg} and \ref{alg:prec_cg_alg} respectively.

If we only consider vectorial operations, Algorithm \ref{alg:cg_alg} requires the computation of three scalar products, three linear combinations of vectors and a matrix-vector multiplication, accounting for a total of $4p + 2n_\bQ$ flops. The preconditioned variant, only requires the additional cost of evaluating $\bM ^{-1}\bz$, which, for Jacobi preconditioning, accounts for $p$ flops.

\begin{minipage}{0.48\textwidth}
\begin{algorithm}[H]
\caption{Conjugate Gradient } \label{alg:cg_alg}
			\textbf{Input: }$\bQ \in \RR^{p \times p}$ symmetric positive definite, $\bb \in \RR ^p$, $\epsilon >0$ desired accuracy.\\
			\textbf{Output:} Approx. solution of $\bQ \bx = \bb$.\\
			\begin{algorithmic}[1]
				\State $\bx_0 = \mathbf{0}$, $\br_0 = \bb$, $\bp _0 = \br _0$.
				\For{$j = 0, 1, \dots$ until $||\br_j||_2 <\epsilon$,}
					\State $\alpha_{j}  = \dfrac{\br _{j}^T \br _{j}}{\bp_j ^T \bQ \bp_j}$
					\State $\bx_{j+1} = \bx_j + \alpha _j \bp _j$
					\State $\br _{j+1} = \br _{j} - \alpha _{j} \bQ \bp _{j}$
					\State $\beta _j  = \dfrac{\br _{j+1} ^T \br _{j+1} }{\br _{j}^T \br _{j}}$
					\State $\bp _{j+1}  = \br _{j+1} + \beta _j \bp _{j}$
				\EndFor \\
				\Return $\bx _{j+1}$\;
	\end{algorithmic}
\end{algorithm} 
\vspace{.05cm}
\end{minipage}
\begin{minipage}{0.5\textwidth}
\begin{algorithm}[H]
\caption{Preconditioned CG } \label{alg:prec_cg_alg}
			\textbf{Input: }$\bQ \in \RR^{p \times p}$ symmetric positive definite, $\bb \in \RR ^p$, $\epsilon >0$, \red{$\bM$ preconditioner}.\\
			\textbf{Output:} Approx. solution of $\bQ \bx = \bb$.\\
			\begin{algorithmic}[1]
				\State $\bx_0 = \mathbf{0}$, $\br_0 = \bb$, \red{$\bz_0 = \bM ^{-1}\br_0$, $\bp _0 = \bz _0$}.
				\For{$j = 0, 1, \dots$ until $||\br_j||_2 <\epsilon$,}
					\State $\alpha_{j}  = \dfrac{\br _{j}^T \red{\bz _{j}}}{\bp_j ^T \bQ \bp_j}$
					\State $\bx_{j+1} = \bx_j + \alpha _j \bp _j$
					\State $\br _{j+1} = \br _{j} - \alpha _{j} \bQ \bp _{j}$
					\State \red{$\bz _{j+1} = \bM^{-1}\br _{j+1}$}
					\State $\beta _j  = \dfrac{\br _{j+1} ^T \red{\bz _{j+1} }}{\br _{j}^T \red{\bz _{j}}}$
					\State $\bp _{j+1}  = \red{\bz _{j+1}} + \beta _j \bp _{j}$
				\EndFor \\
				\Return $\bx _{j+1}$\;
	\end{algorithmic}
\end{algorithm} 
\end{minipage}

\subsection{Polya-Gamma augmented Gibbs sampler}
 \begin{algorithm}[H]
   \begin{algorithmic}
\State \textbf{Input:} hyperparameters $\bm_0,\ \alpha_k,\ \bPhi_k$; vector of observations $\by$; design matrix $\X$; initial values $\bth ^{(0)},\ \bOmega^{(0)},\ \bT_k ^{(0)}$

     \For{$t = 1:T$}
     \State Sample 
 $\bth ^{(t)}\sim p(\bth \mid \by, \bV , \bOmega ^{(t-1)}, \{\bT_k ^{(t-1)}\}_{k=1}^K)$ according to \eqref{eq:gibbs_theta}
	\State Sample 
 $\omega _i  ^{(t)}\sim p(\omega_i \mid \by, \bV, \bth ^{(t)}) $, for each $i = 1,\dots, \n$, according to \eqref{eq:gibbs_omega}
 	\State Sample 
 $\bT_k  ^{(t)}\sim p(\bT_k \mid \by, \bV, \bth ^{(t)})  $, for each $k=1,\dots, K$, according to \eqref{eq:gibbs_prec}
   \EndFor
   \\ \vspace{0.2cm}  \textbf{Output:} MCMC samples $\{ (\bth^{(t)},\ \bOmega^{(t)},\ , \{\bT_k ^{(t)}\}_{k=1}^K); \ t = 0,\dots, T\} $.
\end{algorithmic}
\caption{PG augmented Gibbs sampler}
\label{alg:gibbs}
 \end{algorithm}

\section{Additional Figures}\label{sec:suppl_figures}
\subsection{Sparsity structure of $\bQ$ resulting from Example \ref{ex:worst_sla}}
\begin{figure}[H]
\centering
\includegraphics[width=.4\textwidth]{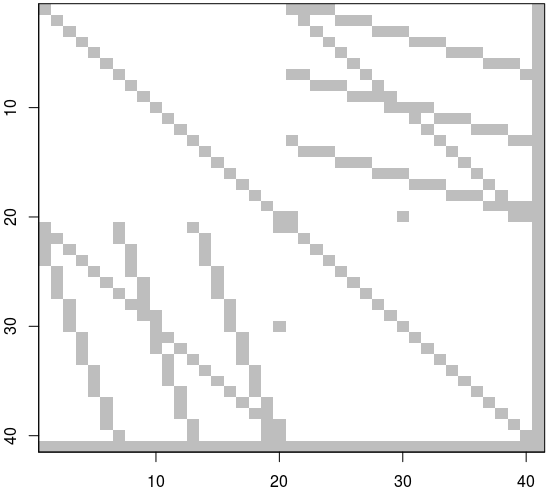}
\caption{Precision matrix $\bQ$ of $\law(\bth \mid \by, \bg)$ induced by the structured design of Example \ref{ex:worst_sla} with $G=20$, $d=3$ and default ordering $(\theta_{1,1}, \dots, \theta_{1,G}, \theta_{2,1}, \dots, \theta_{2,G}, \theta_0)$.}
\label{fig:bad_design}
\end{figure}

\subsection{Graphical representation of the graph in Example \ref{ex:pairwise_non_connected}}
\begin{figure}[H]
  \centering
        \includegraphics[scale=.45]{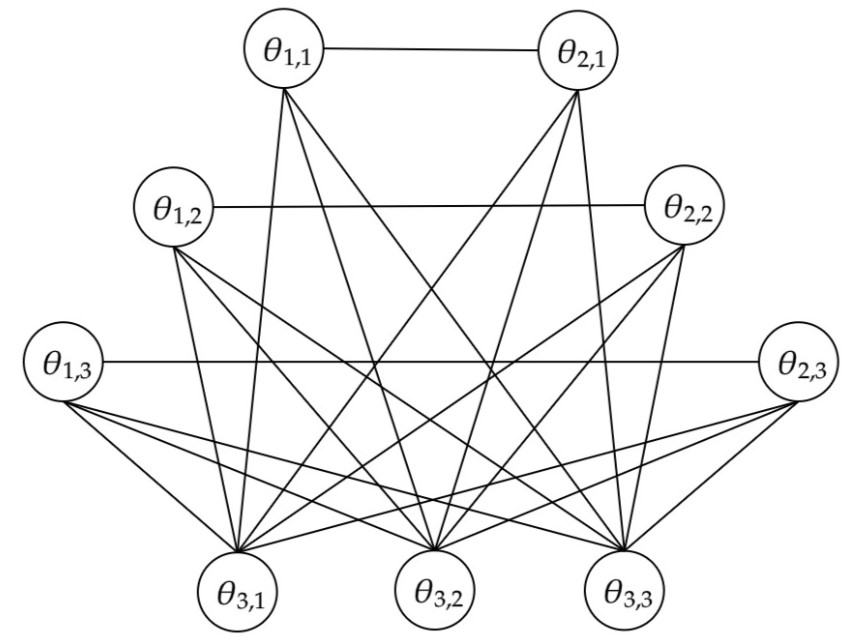}
\caption{Conditional independence graph of Example \ref{ex:pairwise_non_connected}. Obtained for $G_1 = G_2 = G_3 = 3$, and $N= 9$. }\label{fig:pairwise_non_connected}
\end{figure}

\end{document}